\documentclass[a4paper,UKenglish,cleveref, thm-restate]{lipics-v2021}

\title{\texorpdfstring{Algebraic Characterizations of Classes of Regular Languages in \(\DynFO\)}{Algebraic Characterizations of Classes of Regular Languages in DynFO}} %
\author{Corentin Barloy}{Ruhr University Bochum, Germany \and \url{https://barloy.name/}}{corentin.barloy@rub.de}{https://orcid.org/0000-0001-5420-8761}{}

\author{Felix Tschirbs}{Ruhr University Bochum, Germany}{felix.tschirbs@rub.de}{}{}

\author{Nils Vortmeier}{Ruhr University Bochum, Germany}{nils.vortmeier@rub.de}{https://orcid.org/0009-0000-2821-7365}{}

\author{Thomas Zeume}{Ruhr University Bochum, Germany}{thomas.zeume@rub.de}{https://orcid.org/0000-0002-5186-7507}{}

\authorrunning{C.~Barloy, F.~Tschirbs, N.~Vortmeier and T.~Zeume} %

\Copyright{Corentin Charlie Luc Barloy, Felix Tschirbs, Nils Vortmeier, Thomas Zeume} %

\ccsdesc[500]{Theory of computation~Complexity theory and logic}
\ccsdesc[500]{Theory of computation~Algebraic language theory}
\ccsdesc[500]{Theory of computation~Regular languages}

\keywords{Dynamic descriptive complexity, formal languages, monoid theory} %

\relatedversion{} %

\funding{All authors are supported by the Deutsche Forschungsgemeinschaft (DFG, German Research Foundation), grant 532727578.}%

\acknowledgements{We thank the STACS reviewers for their careful reading and insightful comments, which significantly improved the quality of this article.}%

\nolinenumbers %

\hideLIPIcs

 \EventEditors{Meena Mahajan, Florin Manea, Annabelle McIver, and Nguy\~{\^{e}}n Kim Th\'{\u{a}}ng}
 \EventNoEds{4}
 \EventLongTitle{43rd International Symposium on Theoretical Aspects of Computer Science (STACS 2026)}
 \EventShortTitle{STACS 2026}
 \EventAcronym{STACS}
 \EventYear{2026}
 \EventDate{March 9--13, 2026}
 \EventLocation{Grenoble, France}
 \EventLogo{}
 \SeriesVolume{364}
 \ArticleNo{77}

\usepackage{graphicx}
\usepackage{cancel}
\usepackage{xspace}
\usepackage{tikz}
\usetikzlibrary{arrows,decorations.pathmorphing,backgrounds,calc,positioning,fit,petri,matrix,backgrounds, decorations.pathreplacing, shapes.geometric}
\usetikzlibrary{calligraphy}

\newcommand{\Member}[1]{\myproblem{Member}(#1)}

\newcommand{\StrictPrefixPb}{\myproblem{StrictPrefix}}
\newcommand{\StrictPrefix}[1]{\myproblem{StrictPrefix}(#1)}

\newcommand{\StrictSuffixPb}{\myproblem{StrictSuffix}}
\newcommand{\StrictSuffix}[1]{\myproblem{StrictSuffix}(#1)}

\newcommand{\regex}[1]{\ensuremath{\mathtt{#1}}}
\newcommand{\up}[1]{\ensuremath{\uparrow #1}}

\newcommand{\tpl}{\bar}

\newcommand{\mtext}[1]{\textsc{#1}}

\newcommand{\set}{\mtext{set}\xspace}

\newcommand{\schema}{\ensuremath{\tau}\xspace}

\newcommand{\aux}{\ensuremath{\calA}\xspace}%

\newcommand{\df}{\ensuremath{\mathrel{\smash{\stackrel{\scriptscriptstyle{
    \text{def}}}{=}}}} \;}

\newcommand  {\myclass} [1]  {\ensuremath{\textsf{\upshape #1}}}

\newcommand{\StaClass}[1]{\myclass{#1}\xspace}

\newcommand{\DynClass}[1]{\myclass{Dyn#1}\xspace}

\newcommand{\UDynClass}[1]{\myclass{UDyn#1}\xspace}

\newcommand  {\myproblem} [1] {\normalfont{\textsc{#1}}\xspace}

\newcommand{\FO}{\StaClass{FO}}

\newcommand{\CQ}[1][]{\StaClass{CQ}}
\newcommand{\UCQ}[1][]{\StaClass{UCQ}}
\newcommand{\CQneg}[1][]{\StaClass{CQ\ensuremath{^{\mneg}}}}
\newcommand{\UCQneg}[1][]{\StaClass{UCQ\ensuremath{^{\mneg}}}}
\newcommand{\Prop}{\StaClass{Prop}}

\newcommand{\mneg}{\neg} %

\newcommand{\DynProp}{\DynClass{Prop}}

\newcommand{\DynFO}{\DynClass{FO}}

\newcommand{\DynC}{\DynClass{\ensuremath{\calC}}}
\newcommand{\DynSt}{\DynClass{\ensuremath{\Sigma_{2}}}}

\newcommand{\UDynProp}{\UDynClass{Prop}}

\newcommand{\UDynC}{\UDynClass{\ensuremath{\calC}}}
\newcommand{\UDynSt}{\UDynClass{\ensuremath{\Sigma_{2}}}}
\newcommand{\UDynSo}{\UDynClass{\ensuremath{\Sigma_{1}}}}

\newcommand{\UDynSop}{\UDynClass{\ensuremath{\Sigma_{1}^+}}}

\newtheorem{fact}[theorem]{Fact}
\theoremstyle{definition}
\newtheorem*{question*}{Question}
\newtheorem*{openquestion*}{Open question}

\newenvironment{proofsketch}{\begin{proof}[Proof sketch.]}{\end{proof}}

\providecommand {\calA}      {{\mathcal A}\xspace}

\providecommand {\calC}      {{\mathcal C}\xspace}

\providecommand {\calL}      {{\mathcal L}\xspace}

\providecommand {\calV}      {{\mathcal V}\xspace}

\newcommand{\cA}{\mathcal{A}}

\newcommand{\cV}{\mathcal{V}}
\newcommand{\bV}{\mathbf{V}}
\newcommand{\bW}{\mathbf{W}}

\newcommand{\bB}{\mathbf{B}}
\newcommand{\bE}{\mathbf{E}}

\newcommand{\bG}{\mathbf{G}}

\newcommand{\bJ}{\mathbf{J}}

\newcommand{\prog}{\ensuremath{\Pi}\xspace}

\newcommand{\auxSchema}{\ensuremath{\schema_{\text{aux}}}\xspace}

\newcommand{\greenH}{\mathcal{H}}
\newcommand{\greenR}{\mathcal{R}}
\newcommand{\greenL}{\mathcal{L}}
\newcommand{\greenJ}{\mathcal{J}}
\newcommand{\greenK}{\mathcal{K}}

\definecolor{iltisBeige1}{HTML}{fef6ee}
\definecolor{iltisBeige2}{HTML}{e3d5c8}
\definecolor{iltisBeige3}{HTML}{cfbeb0}
\definecolor{iltisBeige4}{HTML}{bfac9b}

\definecolor{iltisGrey1}{HTML}{edebe8}
\definecolor{iltisGrey2}{HTML}{d9d6d2}
\definecolor{iltisGrey3}{HTML}{c2bfbc}
\definecolor{iltisGrey4}{HTML}{a6a4a1}

\definecolor{iltisLightGreen1}{HTML}{f4ffe9}
\definecolor{iltisLightGreen2}{HTML}{d4f7b2}
\definecolor{iltisLightGreen3}{HTML}{bde697}
\definecolor{iltisLightGreen4}{HTML}{a3d177}

\definecolor{iltisGreen1}{HTML}{cfffe1}
\definecolor{iltisGreen2}{HTML}{a2e8bd}
\definecolor{iltisGreen3}{HTML}{86d1a2}
\definecolor{iltisGreen4}{HTML}{6fbf8d}

\definecolor{iltisYellow1}{HTML}{fef2d0}
\definecolor{iltisYellow2}{HTML}{ffe3a2}
\definecolor{iltisYellow3}{HTML}{ffd77d}
\definecolor{iltisYellow4}{HTML}{f2c55e}

\definecolor{iltisRed1}{HTML}{ffe9e6}
\definecolor{iltisRed2}{HTML}{eda498}
\definecolor{iltisRed3}{HTML}{e08475}
\definecolor{iltisRed4}{HTML}{c26e60}

\definecolor{iltisOrange1}{HTML}{ffdfb3}
\definecolor{iltisOrange2}{HTML}{ffc97d}
\definecolor{iltisOrange3}{HTML}{ebb467}
\definecolor{iltisOrange4}{HTML}{e0a34c}

\definecolor{iltisCyan1}{HTML}{e0fffe}
\definecolor{iltisCyan2}{HTML}{b4e0df}
\definecolor{iltisCyan3}{HTML}{95c7c5}
\definecolor{iltisCyan4}{HTML}{81b3b0}

\definecolor{iltisBlue1}{HTML}{cce8ff}
\definecolor{iltisBlue2}{HTML}{8db8d9}
\definecolor{iltisBlue3}{HTML}{6f9abd}
\definecolor{iltisBlue4}{HTML}{5c86a8}
\definecolor{iltisBlue5}{HTML}{2e5b80}

\definecolor{iltisViolet1}{HTML}{ede6ff}
\definecolor{iltisViolet2}{HTML}{b1a5cc}
\definecolor{iltisViolet3}{HTML}{9b8eba}
\definecolor{iltisViolet4}{HTML}{8578a6}

\colorlet{midgray}{black!50}

\tikzstyle{dEdge}=[
  -latex', %
  thick, 
  shorten >=2pt, 
  shorten <=2pt,
  draw=black!80,
]

\tikzstyle{dhEdge}=[
  -latex', %
  thick, 
  shorten >=3pt, 
  shorten <=3pt,
  draw=black!80,
]
\tikzstyle{uEdge}=[
  thick, 
  shorten >=3pt, 
  shorten <=3pt,
  draw=black!80,
]
\tikzstyle{uhEdge}=[
  thick, 
  shorten >=3pt, 
  shorten <=3pt,
  draw=black!80,
]

\tikzstyle{cEdge}=[
  ultra thick, 
  shorten >=3pt, 
  shorten <=3pt,
  draw=black!80,
]

\tikzstyle{dotsEdge}=[
  very thick, 
  loosely dotted, 
  shorten >=3pt, 
  shorten <=3pt
]

\tikzstyle{snakeEdge}=[
  ->, 
  decorate, 
    shorten >=2pt, 
  shorten <=2pt,
  decoration={snake,amplitude=.4mm,segment length=2.5mm,post length=0.0mm},
]

\tikzstyle{snakeEdgea}=[
  ->, 
  decorate, 
  decoration={snake,amplitude=.4mm,segment length=3mm,post length=0.5mm}
]
\tikzstyle{blackNode}=[
	shape=circle, draw, fill,inner sep=0pt,minimum size=4pt
]
 
\begin{document}

\maketitle

\begin{abstract}
	
	This paper explores the fine-grained structure of classes of regular languages maintainable in fragments of first-order logic within the dynamic descriptive complexity framework of Patnaik and Immerman. A result by Hesse states that the class of regular languages is maintainable by first-order formulas even if only unary auxiliary relations can be used. Another result by Gelade, Marquardt, and Schwentick states that the class of regular languages coincides with the class of languages maintainable by quantifier-free formulas with binary auxiliary relations.
	
	We refine Hesse's result and show that with unary auxiliary data $\exists^*\forall^*$-formulas can maintain all regular languages. We then obtain precise algebraic characterizations of the classes of languages maintainable with quantifier-free formulas and positive $\exists^*$-formulas in the presence of unary auxiliary relations.

\end{abstract}

\newpage

\section{Introduction}
\label{sec:intro}
In seminal work by Schützenberger as well as McNaugthon and Papert, the class of regular languages has been characterized algebraically as the class of languages with finite syntactic monoids~\cite{Schutzenberger_syntactic_1955} and logically as the class of languages definable by monadic second-order sentences~\cite{buchi_mso_reg_1960}.
In subsequent work, a multitude of  algebraic and logical characterisations have been obtained for subclasses of regular languages,
including a characterization of the class of star-free regular languages (languages describable by regular expressions without Kleene star but with negation) as the class of languages with finite aperiodic monoids and definable in first-order logic~\cite{schutzenberger_aperiodic_1965,mcnaughton_fo_star_free_1971};
and the characterization of piecewise testable languages  (languages describable by the set of subwords appearing in the word) as languages with $\mathcal J$-trivial syntactic monoids and definable by first-order formulas without quantifier alternation~\cite{simon_j_1975}.

Here we are interested in algebraic properties of formal languages defined in the dynamic descriptive complexity framework of Patnaik and Immerman \cite{PatnaikI97}.  In this framework, strings are subject to changes, and the goal is to maintain membership within a language via logical formulas, possibly using auxiliary information stored in auxiliary relations (which need to be updated by logical formulas as well). An important class of languages in this setting are the languages maintainable via first-order formulas, called $\DynFO$. It has been shown early on that $\DynFO$ contains all regular and even all context-free languages \cite{PatnaikI97, GeladeMS12}.

The goal of this paper is to initiate a search for connections between algebra and dynamic descriptive complexity classes. This is in a similar spirit as the search for connections between algebra and logics initiated by Schützenberger, McNaugthon and Papert for the static realm, as sketched above. Our guiding question is:
\smallskip
\begin{itemize}
 \item[] \emph{Are there restrictions of $\DynFO$ that lead to natural subclasses of the regular languages with nice algebraic characterizations?}
\end{itemize}
\smallskip
For answering this question, it is natural to first identify restrictions of $\DynFO$ that can still maintain all regular languages. This has been studied in prior work. A result by Hesse states that the class of regular languages is maintainable by first-order formulas even if only very restricted auxiliary information, namely only unary relations, can be used~\mbox{\cite[Theorem 2.8]{hesse_phd_2003}}\footnote{A variant of this result has already been reported in \cite{PatnaikI97}, but additional built-in arithmetic has been used as auxiliary information.}. A seminal result by Gelade, Marquardt, and Schwentick states that, if binary auxiliary relations can be used, then the class of regular languages coincides with the class $\DynProp$ of languages maintainable by quantifier-free formulas \cite[Theorem 3.2]{GeladeMS12}. 

We refine the result by Hesse. Denote by \(\UDynSt\) the class of languages maintainable by $\exists^*\forall^*$-formulas using unary auxiliary relations.
\begin{theorem} \label{thm:thm1}
  All regular language are in \(\UDynSt\).
\end{theorem}
This result can be obtained by close inspection of Hesse's proof, which maintains properties of the execution graph of finite state automata. Here, we provide an algebraic reformulation of the proof relying on Green's relations (see Section \ref{sec:udynst}).

Thus, when looking for restrictions of $\DynFO$ that correspond to natural subclasses of the regular languages with nice algebraic characterizations, one can explore (a) fragments of (binary) $\DynProp$, or (b) fragments of  \(\UDynSt\). 

We obtain results along both directions. For binary $\DynProp$, it is natural to further restrict the auxiliary relations. We precisely characterize the languages in $\DynProp$ with unary auxiliary relations  by algebraic properties of their syntactic monoids (see Section \ref{sec:udynprop}).

\begin{theorem}\label{thm:thm2}
 	A regular language is in $\UDynProp$ if any only if its syntactic monoid is a group.
\end{theorem}

For $\UDynSt$, it is natural to further restrict the syntactic structure of formulas used for updates. We precisely characterize the languages maintainable by positive $\exists^*$-formulas, denoted $\UDynSop$, by the algebraic properties of their ordered syntactic monoids (see Section~\ref{sec:udynsop}).

\begin{theorem}\label{thm:thm3}
	A regular language is in $\UDynSop$ if and only if its ordered syntactic monoid is in $\bJ^{+} \ast \bG$.
\end{theorem}
Here, $\bJ^{+} \ast \bG$ denotes the wreath product of the classes $\bJ^{+}$ and $\bG$. 

In the proofs of both Theorem \ref{thm:thm2} and Theorem \ref{thm:thm3}, the ``if'' direction requires to provide update formulas for all regular languages whose (ordered) syntactic monoids satisfy some property and is standard. The ``only-if'' direction requires to prove lower bounds against $\DynProp$ (resp.  $\UDynSop$) for all non-group languages (resp. all languages whose ordered syntactic monoid is not in $\bJ^{+} \ast \bG$). To this end we prove a lower bound for a single language using variants of the Substructure lemma \cite{GeladeMS12} (see also \cite{Zeume15thesis}) and lift this lower bounds to all languages without the property using algebraic insights.

We leave open the question to find a precise algebraic characterization of languages in $\UDynSo$, i.e., languages maintainable by (not necessarily positive) $\exists^*$-formulas and unary auxiliary relations. We discuss challenges towards answering this question in Section \ref{sec:conclusion}.

\section{Preliminaries}
\label{sec:prelim}

In this section, also to fix notation, we recall basics of dynamic descriptive complexity and algebraic formal language theory. 
\subsection{Dynamic descriptive complexity}

In this paper, we are interested in the dynamic membership problem for fixed regular languages. Let $L \subseteq \Sigma^*$ be a language over some alphabet $\Sigma$. 
The input structure for the dynamic membership problem $\Member{L}$ of $L$ is an encoding of a word $w = w_1 \cdots w_{n}$ as a relational structure over the domain $\{1, \ldots, n\}$ of the positions of the word. This structure has (1) for every $\sigma \in \Sigma$ a unary relation $W_\sigma$ for storing the positions of $w$ that contain $\sigma$ and (2) a binary relation $\leq$ representing the linear order of the positions of $w$. 
Following Patnaik and Immerman as well as Gelade et al. \cite{PatnaikI97, GeladeMS12}, each $w_i$ is either in $\Sigma$ or $\epsilon$, i.e., each position $i \in [n]$ can occur in at most one relation $W_\sigma$. For the sake of simplicity, we denote both input words and their encoding by the same symbol $w$.

In the dynamic membership problem, the input structure can be changed by operations of the form $\set_\sigma(y)$, for $\sigma \in \Sigma \cup \{\epsilon\}$. A concrete change $\set_\sigma(i)$ instantiates $y$ and sets the symbol at position $i$ to $\sigma$. The dynamic membership problem $\Member{L}$ asks whether the current input structure encodes a word in $L$, i.e. whether $w_1 \cdots  w_{n} \in L$.%

We will explore the resources needed for maintaining the dynamic membership problem within the dynamic descriptive complexity framework of Patnaik and Immerman \cite{PatnaikI97}.

A \emph{dynamic program} $\prog$ stores the input structure as well as a set~$\calA$ of auxiliary relations over some (relational) schema $\auxSchema$ and over the same domain as the input word. For every auxiliary relation symbol $R \in \auxSchema$ and every change operation $\set_\sigma$, the program $\prog$ has an \emph{update formula} $\varphi_\sigma^R(\tpl x;y)$, which can access input and auxiliary relations. Whenever the current input word $w$ is changed by $\set_\sigma(i)$ to the new input word $w'$, the new auxiliary relation $R^{\aux'}$ in the updated set $\aux'$ consists of all tuples $\tpl a$ such that $\varphi_\sigma^A(\tpl a;i)$ is satisfied in the structure $(w', \aux)$.

A dynamic program $\prog$ maintains the dynamic membership problem $\Member{L}$ if it has a distinguished auxiliary bit $q$ (i.e., a $0$-ary auxiliary relation), which always indicates whether the current input word is in $L$. More precisely, following Patnaik and Immerman, we assume that the initial input word $w_0$ is empty and the initial auxiliary structure $\aux_0$ is initialized by a first-order formula depending on $w_0$. The program $\prog$ \emph{maintains} $\Member{L}$, if after changing $w_0$ to \(w'\) by any sequence of changes and subsequently applying the corresponding update formulas starting from $(w_0, \aux_0)$, the obtained bit $q$ is true if and only if $w' \in L$.

The dynamic problem $\Member{L}$ is in the class \DynFO if it can be maintained by a dynamic program with \FO update formulas and an \FO initialization formula. It is in $k$-ary \DynFO if all auxiliary relations are at most $k$-ary. For a class \(\calC\)  of formulas, it is in \(\DynC\) if all update formulas are in \(\calC\) and the initialization formula is in \(\FO\).

Throughout this paper we assume that formulas are in prenex normal form. Among the classes $\calC$ we study are the class \(\Prop\) of quantifier-free formulas; the class \(\Sigma_{1}\) of $\exists^*$-formulas  (i.e., formulas in prenex form with only existential quantifiers); and the class \(\Sigma_{2}\) of $\exists^* \forall^*$-formulas  (i.e., formulas in prenex form starting with a block of existential quantifiers followed by a block of universal quantifiers). We also consider the variants \(\Prop^{+}\), \(\Sigma_{1}^{+}\) and \(\Sigma_{2}^{+}\) where negations may only be used directly in front of the linear order symbol $\leq$. A focus will be on programs that can only use unary auxiliary relations, in which case we denote classes by \(\UDynC\) for classes \(\calC\) of update formulas.

\begin{example}\label{example:aa}
We show that the regular language \(\regex{\Sigma^{*}aa\Sigma^{*}}\) over \(\Sigma = \{a,b\}\) is in \(\UDynSo\), so, can be maintained using only unary auxiliary relations and update formulas in prenex form that only use existential quantifiers.

Our goal is to maintain a unary relation \(N\) that stores all positions $j$ such that the next letter to the right of $j$ that does not contain $\epsilon$ is an $a$.
This relation can directly be maintained using both \(\Sigma_{1}\) and \(\Pi_{1}\) update formulas. As \(\Pi_{1}\) update formulas are not allowed, we store a slightly differently defined relation $N'$ and an additional bit $c$.
If the flag $c$ is true, the relation $N'$ is equal to $N$. If $c$ is false, then $N'$ is the complement of $N'$, that is, contains all position $j$ such that either all positions to the right of $j$ are $\epsilon$ or the first positions that is not $\epsilon$ is a $b$.
Using this relation, we can obtain \(N\) as \(N(j) = (\neg c \land \neg N'(j)) \lor (c \land N'(j))\).

We now describe the update formulas for maintaining \(N'\) and $c$. We additionally use unary auxiliary relations $W^o_\sigma$ to store the old input word before the change. These are trivial to maintain.
	\begin{itemize}
	\item In case of a change of the form \(\set_{b}(i)\), we set $c = 1$ and the relation $N'$ after the update is equal to $N$. 
	A position $j$ is in $N'$ if $j \in N$ before the change and either $i \leq j$ or there is a position between $i$ and $j$ that is an $a$. The corresponding update formula for $N'$ is 
	$\varphi_b^{N'}(j;i) = \exists k \big[ N(j) \wedge \big(i \leq j \vee (j < i \wedge j < k \wedge k < i \wedge W_a(k)) \big) \big]$, where $x < y$ is an abbreviation for $x \leq y \wedge \neg(x = y)$.
	\item In case of a change of the form \(\set_{a}(i)\), we set $c = 0$ and the relation $N'$ after the update is the complement of $N$. 
	A position $j$ is in $N'$ if $j \not\in N$ before the change and either $i \leq j$ or there is a position between $i$ and $j$ that is a $b$. The corresponding update formula for $N'$ is 
	$\varphi_a^{N'}(j;i) = \exists k \big[ \neg N(j) \wedge \big(i \leq j \vee (j < i \wedge j < k \wedge k < i \wedge W_b(k)) \big) \big]$.
	\item In case of a change \(\set_{\epsilon}(i)\), the update depends on the old symbol at position~$i$. 
	If the symbol was an $a$, the relation $N'$ is equal to $N$ after the update and we set $c=1$. If the removed $a$ was the first letter to the right of $j$, then we have to check the first letter to the right of $i$. Otherwise, the change is not relevant for position $j$.
	So, if $j < i$, the next letter from \(j\) is an \(a\) if and only if \(N(j)\) and \(N(i)\) were both true before the change, or there is an $a$ between \(j \) and \(i\) and \(N(j)\) was already true.
	
	If the change replaces $b$ with $\epsilon$, the relation $N'$ is the complement of $N$ after the update and we set $c=0$. If $j < i$, the next letter from \(j\) is not an \(a\) if and only if \(N(j)\) and \(N(i)\) were both false before the change, or there is a $b$ between \(j \) and \(i\) and \(N(j)\) was false.
	The update formula is
	\begin{align*}
		& \varphi_\epsilon^{N'}(j;i) = \exists k \Big[ \\
			& W^o_a(i) \wedge \Big((i \leq j \wedge N(j)) \vee \big(j < i \wedge N(j) \wedge (N(i) \vee ( j < k \wedge k < i \wedge W_a(k)) ) \big)  \Big) \vee \\ 
		 & W^o_b(i) \wedge \Big((i \leq j \wedge \neg N(j)) \vee \big(j < i \wedge \neg N(j) \wedge (\neg N(i) \vee ( j < k \wedge k < i \wedge W_b(k)) ) \big)  \Big) \Big].
	\end{align*}
\end{itemize}
	Similarly, we can maintain a unary auxiliary relation \(P\) that contains all positions $j$ such that the next letter to the left of $j$ that does not contain $\epsilon$ is an $a$.
	
	We now describe the maintenance of the bit \(q\) that is set to \(1\) whenever \(w\) is in \(L= \Sigma^{*}aa\Sigma^{*}\).
	\begin{itemize}
		\item In case of a change \(\set_{a}(i)\), the new word \(w\) is in \(L\) if and only if \(q\) or one of \(N(i)\) and \(P(i)\) was already true before the change.
		\item In case of a change \(\set_{b}(i)\), we have \(w \in L\) if and only if there exists a position \(k> i\) with an \(a\) such that \(N(k)\) holds, or there exists a position \(k< i\) with an \(a\) such that \(P(k)\) holds.
		\item In case of a change \(\set_{\epsilon}(i)\) it holds \(w \in L\) if and only if there exists a position \(k> i\) with an \(a\) such that \(N(k)\) holds, or there exists a position \(k< i\) with an \(a\) such that \(P(k)\) holds, or both \(N(i)\) and \(P(i)\) hold.
	\end{itemize}
	These conditions can easily be described by existential update formulas. \lipicsEnd
\end{example}

During our study, we will need the following closure properties.
For a language $L$ over alphabet $\Sigma$ and a letter $\sigma \in \Sigma$, we define \(L \sigma^{-1} = \{ w \mid w\sigma\in L\}\) and \( \sigma^{-1}L = \{ w \mid \sigma w\in L\}\).
\begin{restatable}{lemma}{invmorphismquotients}
\label{lem:inv_morphism_quotients}
	Let $L$ be a language over alphabet $\Sigma$ and $\DynC \in \{\UDynProp, \UDynSop\}$. 
	If $\Member{L}$ can be maintained in $\DynC$, then so can
	\begin{enumerate}
		\item $\Member{h^{-1}(L)}$ %
		for any mapping $h : \Gamma \to \Sigma^*$ and any alphabet $\Gamma$, %
		\item $\Member{L \sigma^{-1}}$ and $\Member{\sigma^{-1}L}$, for any $\sigma \in \Sigma$.
	\end{enumerate}
\end{restatable}

\subsection{Monoids}
For \(\DynFO\) and many of its natural restrictions that possess favourable closure properties, we follow Skovberg Frandsen, Miltersen, and Skyum~\cite{Frandsen_dynamic_1997} in studying dynamic problems for monoids rather than languages. A \emph{monoid} $(M, \cdot)$ is a set $M$ equipped with a binary operation~$\cdot$ that is associative and has an identity element, i.e., the operation satisfies \((x\cdot y)\cdot z=x\cdot (y\cdot z)\) for all elements \(x,y,z\in M \) and there is an element \(1 \in M\) such that \(1\cdot x = x\cdot 1=x\) for all \(x\in M\).
A monoid is a \emph{group} if every element \(x\) has an inverse, i.e.~an element \(y\) such that \(x\cdot y = y\cdot x = 1\).
An element \(e\) is said to be \emph{idempotent} if \(e\cdot e = e\).
For an integer \(n>0\) and \(x\in M\), the \(n\)-fold product of \(x\) is denoted by \(x^{n}\).
It is well known that every element \(x\in M\) has a unique power that is idempotent~\cite[Prop. II.6.31]{pin_book_2014} called \emph{idempotent power} of \(x\) and denoted by \(x^{\omega}\).

A \emph{morphism} \(\varphi: M \rightarrow N\) between two monoids is a mapping such that \(\varphi(1) = 1\) and \(\varphi(x\cdot y) = \varphi(x) \cdot \varphi(y)\) for every \(x,y\in M\).
A \emph{submonoid} \(N\) of a monoid \(M\) is a subset of \(M\) that is a monoid, i.e., it contains $1$ and is closed under $\cdot$, meaning \(x\cdot y\in N\) for all \(x,y\in N\).
A monoid \(N\) is a \emph{quotient} of \(M\) if there exists a surjective morphism from \(M\) to \(N\).
We say that \(N\) \emph{divides} \(M\) if \(N\) is the quotient of a submonoid of \(M\).
A language \(L\) is \emph{recognized} by a monoid \(M\) if there is a morphism \(\varphi:\Sigma^{*} \rightarrow M\) and \(P\subseteq M\) such that \(L= \varphi^{-1}(P)\).
The regular languages are precisely those recognized by a finite monoid (see for instance~\cite[Theorem IV.3.21]{pin_book_2014}).

For a monoid \(M\), we consider dynamic programs over the alphabet \(M\) and define the problem of evaluating a sequence of elements in \(M\).
We say that a dynamic program \(\Pi\) maintains the dynamic membership problem \(\Member{M}\) if
there are \(|M|\) dedicated auxiliary bits that respectively store for every \(x\in M\) whether the maintained word evaluates to \(x\).
Similarly the dynamic problem \(\StrictPrefixPb\) (resp. \(\StrictSuffixPb\)) is maintained via one dedicated unary relation for each \(x\in M\) that stores the value \(1\leq i\leq n\) if the prefix
\(w_{1} \cdot \ldots \cdot w_{i-1}\)  (resp. suffix $w_{i+1} \cdot \ldots \cdot w_n$) evaluates to \(x\).
In this setting, the symbol \(\epsilon\) and the neutral element \(1\) of \(M\) can be used interchangeably.
For instance, the initial word \(w_{0}\) can either be the empty word or \(1^{n}\).

An important monoid is the monoid of functions \(f: Q \rightarrow Q\) for a set \(Q\).
The operation is the composition and the neutral element is the identity function.
Given a deterministic finite state automaton with transition function \(\delta:\Sigma^{*} \rightarrow (Q\rightarrow Q)\), its \emph{transition monoid} is the image \(\delta(\Sigma^{*})\).
The \emph{syntactic monoid} of a regular language is the transition monoid of its minimal automaton.
It is the smallest monoid that recognizes the language and reflects many properties of its regular language.
Indeed, the syntactic monoid possesses more information than its regular language, as exposed by the following claim.
A class of formulas \(\calC\) is closed under \(\lor\) (resp.\ \(\land\), resp.\ \(\neg\)) if for every \(\varphi,\psi\in \calC\), we also have \(\varphi \lor \psi\) (resp.\ \(\varphi\land \psi\), resp.\ \(\neg\varphi\)) in \(\calC\).
All classes under inspection in this paper possesses this property.

\begin{restatable}{fact}{monoidtolanguage}
  \label{claim:monoid_to_language}
  Let \(\calC\) be a class of formulas closed under \(\lor\).
  Let \(L\) be a regular language with syntactic monoid \(M\) and $k$ a number, then
  \[\Member{M} \in k\textit{-ary } \DynC \Rightarrow \Member{L} \in k\textit{-ary } \DynC.\]
\end{restatable}

We will see in~\cref{sec:udynprop} and~\cref{sec:udynsop} that when the class of formulas under consideration is well-behaved, the converse of this claim is also true.

We conclude this section by stating that maintainability of \(\Member{M}\) for a monoid implies maintainability of \(\Member{N}\) for every \(N\) that divides \(M\).

\begin{restatable}{fact}{closuredivision}
  \label{fact:closure_division}
  Let \(\calC\) be a class of formulas closed under \(\lor\) and $k$ any number. Let \(M\) be a monoid such that \(\Member{M}\in k\textit{-ary }\DynC\).
  Then for every \(N\) that divides \(M\) we have \(\Member{N}\in k\textit{-ary }\DynC\).
\end{restatable}

\section{\texorpdfstring{The regular languages of \(\UDynSt\)}{The regular languages of UDynSigma-2}}
\label{sec:udynst}
In this section we prove that every regular language can be maintained by $\Sigma_2$-formulas and only unary relations. That is, we prove the following.

\begin{theorem}[Restatement of \cref{thm:thm1}]
  \label{thm:reg_in_udynst}$\;$ \\
  For any regular language \(L\): $\;\;$ \(\Member{L}\) is in \(\UDynSt\).
\end{theorem}

The proof is algebraic by nature, hence we show that \(\Member{M}\) for a syntactic monoid \(M\) can be maintained and then apply~\cref{claim:monoid_to_language}.
The approach towards proving Theorem~\ref{thm:reg_in_udynst} is as in Hesse~\cite{hesse_phd_2003}: We want to exhibit maintainable unary relations that allow to compute the evaluation of every infix of a word from \(M^{*}\). This is formalized in the following lemma.

\begin{restatable}{lemma}{strelationsexists}
  \label{claim:st_relations_exists}
  Let $M$ be a finite monoid. There is a dynamic program $\prog$ that maintains unary auxiliary relations over a schema $\schema_M$ such that
  \begin{enumerate}[i)]
  	\item $\prog$ only uses $\Sigma_2$ update formulas, and
    \item for every \(x\in M\) there is a $\Sigma_2$ formula \(\psi_{x}(j,k)\) over the schema of $\prog$ that is satisfied if and only if the infix \(w_{j+1}\cdots w_{k-1}\) of the input word $w$ evaluates to \(x\).
  \end{enumerate}
\end{restatable}

The definition of the auxiliary relations and the proof of the lemma  relies on the local theory of monoids, as developed by Green~\cite{green_relations_51}. We first show that
Theorem~\ref{thm:reg_in_udynst} follows.

\begin{proof}[Proof of~\cref{thm:reg_in_udynst}]
    We prove that \(\Member{M}\) is in \(\UDynSt\) for the syntactic monoid of a regular language. As $\UDynSt$ is closed under disjunction, this implies the result due to~\cref{claim:monoid_to_language}.  Towards this end, let $\prog$ be the dynamic program and $\psi_x$ the formulas expressing the evaluation of infixes of the input word, as given by~\cref{claim:st_relations_exists}.
	Note that the formulas $\psi_x$ give the evaluation of the input word ahead of a change. To maintain $\Member{M}$, we need the evaluation after a change.  Suppose a change $\set_{\sigma}(i)$ occurs. The formula 
	  \begin{align*}
		    \psi_{x}^{\sigma}(j,k;i) & = \Big[(j<i<k) \land  \Big( \bigvee_{\stackrel{y,z\in M}{y \sigma z = x}} \psi_{y}(j,i) \land \psi_{z}(i,k) \Big)\Big]  \lor \Big[(i\leq j \lor k\leq i) \land \psi_{x}(j,k)\Big]
		  \end{align*}
	then expresses that the infix from position $j+1$ to position $k-1$ of the changed input word evaluates to $x$. As the formulas $\psi_x, \psi_y, \psi_z$ are not in the scope of a negation and no further quantifiers are used, the formula is also in $\Sigma_2$.
	
	Now, the formula
	\[ \varphi(i) = \exists j,k \Big[ \min(j) \land \max(k) \land  \Big( \bigvee_{\stackrel{v,y,z\in M}{yvz = x}} W_y(j) \land W_z(k) \land \psi_{v}^\sigma(j,k;i) \Big)\Big] ,\]
	with $\min(j) = \forall i ( i\leq j \rightarrow i=j)$ and  $\max(k) = \forall i ( i\geq k \rightarrow i=k)$ %
	expresses that \(w\) evaluates to \(x\) after the change on position \(i\), for any fixed \(x\in M\).
    This formula still has only one alternation of quantifiers and belongs to \(\Sigma_{2}\).
    Together with the update formulas of $\prog$, we have a dynamic $\DynSt$ program for $\Member{M}$.
  \end{proof}

In the remainder of this section, we recall Green's relations and use them to define the unary relations contained in the schema $\schema_M$ of the program $\Pi$ from  \cref{claim:st_relations_exists}. The proof that these relations satisfy conditions (i) and (ii) from \cref{claim:st_relations_exists} is delegated to the full version of this paper.%

\subparagraph*{Green's relations} The local theory of finite monoids studies the structure of monoids by considering properties of sets of elements that behave the same way.
It is based on the so-called Green's relations, which are a generalization of the notion of divisibility for any monoid.
Over the integers, the divisibility relation is a partial order, whereas Green's relations for arbitrary monoids are preorders, that is, reflexive and transitive relations.
Fix a monoid \(M\).
Its \emph{Green's preorders} are defined by, for \(x,y\in M\):
\begin{itemize}
  \item \(x\leq_{\greenR}y\) whenever there exists \(\alpha\in M\) such that \(x=y\alpha\),
  \item \(x\leq_{\greenL}y\) whenever there exists \(\alpha\in M\) such that \(x=\alpha y\),
  \item \(x\leq_{\greenJ}y\) whenever there exists \(\alpha,\beta\in M\) such that \(x=\alpha y\beta\),
  \item \(x\leq_{\greenH}y\) whenever \(x\leq_{\greenR}y\) and \(x\leq_{\greenL}y\).
\end{itemize}

The \emph{Green's classes} \(\greenR\), \(\greenL\), \(\greenJ\) and \(\greenH\) of \(M\) are the equivalence classes of the preorders.
For instance, two elements \(x,y\in M\) are equivalent for \(\greenR\), denoted by \(x\greenR y\), if and only if \(x\leq_{\greenR}y\) and \(y\leq_{\greenR}x\).
It follows from the definitions that \(\greenH \subseteq \greenR \subseteq \greenJ\) and \(\greenH \subseteq \greenL \subseteq \greenJ\).
Note that for any green relation \(\greenK\) among the four, \(\leq_{\greenK}\) induces a partial order on the set of \(\greenK\)-classes.
The order \(\leq_{\greenJ}\) is of particular interest: assume there is an element \(x\) in some \(\greenJ\)-class \(J\).
Then, for any \(y,z\in M\), \(yxz\) belongs to a \(\greenJ\)-class \(I\) with \(I\leq_{\greenJ}J\).
Therefore, if a factor of a word \(w_{1}\cdots w_{n}\in M^{*}\) evaluates to some element of a  \(\greenJ\)-class $J$, the whole word evaluates to an element from a class $J'$ with $J' \leq_{\greenJ} J$.
There is an analogous statement for \(\greenR\)-classes and \(\greenL\)-classes that we will use repeatedly throughout the paper:
if an element \(x\) is a prefix of another element \(y\), that is there exists \(\alpha\) such that \(x\alpha=y\), then we have \(y\leq_{\greenR} x\).
Likewise, if \(x\) is a suffix of \(y\), so there exists \(\alpha\) such that \(\alpha x=y\), then \(y\leq_{\greenL}x\).

\begin{example}
  The syntactic monoid of \(\regex{(aa)^{\ast}}\) is \(\mathbb{Z}_{2}\), the cyclic group with two elements \(1\) and \(a\) such that \(a\cdot a = 1\).
  Both elements are \(\greenH\)-equivalent (and therefore also \(\greenR\)-, \(\greenL\)-, and \(\greenJ\)-equivalent).
  In fact, in groups, all elements are always \(\greenH\)-equivalent.
\end{example}

\begin{example}
  The syntactic monoid of \(\regex{(a+b)^{*}a(a+b)^{*}}\) is usually denoted by \(U_{1}\).
  It has two elements \(1\) and \(a\) such that \(a\cdot x = x\cdot a = a\) for every element \(x\).
  The two elements are in distinct classes for all of Green's relations, and \(1 \geq_{\greenJ} a\).
\end{example}

\begin{example}
  The syntactic monoid of \(\regex{(a+b)^{*}a}\) is usually denoted by \(U_{2}\).
  It has three elements \(1\), \(a\) and \(b\) such that \(a\) and \(b\) are idempotent and \(a\cdot b = b\) and \(b\cdot a = a\).
  The elements \(a\) and \(b\) are \(\greenR\)-equivalent (and therefore also \(\greenJ\)-equivalent), but not \(\greenL\)-equivalent.
  The neutral element \(1\) is alone in its \(\greenJ\)-class which is \(\greenJ\)-smaller than the one of \(a\) and \(b\).
\end{example}

\subparagraph*{Unary auxiliary relations for expressing evaluations of infixes} 
The proof of~\cref{claim:st_relations_exists} relies on the fact that finite monoids have only finitely many \(\greenJ\)-classes.
Therefore, when we evaluate a word from left to right, there is only a finite number of changes of \(\greenJ\)-classes.
We make that intuition explicit. In the following, for a word \(w = w_1\cdots w_n \in M^{*}\), we write \(w[j,k]\) to denote the subword $w_j \cdots w_k$ of \(w\) between indices \(j\) and \(k\).
For an element  \(x\), we denote by \(\greenJ(x)\) its \(\greenJ\)-class.
When it is clear from the context, we identify a word in \(M^{*}\) by its evaluation.
The following lemma is about the evaluation of a word \(w\in M^{*}\) from left to right, that is evaluating \(w_{1}\), \(w_{1}w_{2}\), \(w_{1}w_{2}w_{3}\), \(\ldots\), \(w_{1}\cdots w_{n}\) in that order.
It states that, except for a constant number of times (depending on \(M\) and not on \(w\)), the \(\greenJ\)-class of the partial evaluations does not change.

\begin{restatable}{lemma}{Jfalling}
  \label{fact:J_falling}
  Let \(w\) be a word from \(M^{*}\) of size \(n\) and \(j \leq n\) a position.
  Then there exist integers \(j=\ell_{0}<\ell_{1}<\cdots<\ell_{m}=n+1\) with \(m\leq |M|\) such that:
  \begin{enumerate}[i)]
    \item for \(0\leq s< m\) and \(\ell_{s}\leq i < \ell_{s+1}\), \(w[j,\ell_{s}]\ \greenJ\ w[j,i]\),
    \item for \(0\leq s < m-1\), \(w[j,\ell_{s+1}-1]\ \cancel{\greenJ}\ w[j,\ell_{s+1}]\).
  \end{enumerate}
  There is a symmetric statement from right to left. That is, there exists integers \(0=\ell_{m}< \cdots <\ell_{1}<\ell_{0}=j\) with \(m\leq |M|\) such that:
  \begin{enumerate}[i)]
    \item for \(0\leq s< m\) and \(\ell_{s+1}< i \leq \ell_{s}\), \(w[\ell_{s},j]\ \greenJ\ w[i,j]\),
    \item for \(0\leq s < m-1\), \(w[\ell_{s+1}-1,j]\ \cancel{\greenJ}\ w[\ell_{s+1},j]\).
  \end{enumerate}
\end{restatable}

The relations promised in~\cref{claim:st_relations_exists} contain the information, for every position \(i\), of the evaluation of the prefixes obtained in the decomposition of~\cref{fact:J_falling} applied to \(i\).
In the following, we slightly abuse notation and write \(x\geq_{\greenJ} J\) to mean that the \(\greenJ\)-class of $x$ is at least \(J\) with respect to the order of \(\greenJ\)-classes.

Let \(w\) be a word from \(M^{*}\), \(i\) a position in the word, \(y\in M\) and \(J\) a \(\greenJ\)-class of \(M\).
Let \(k\) be the greatest position such that \(y\cdot w[j+1,k] \geq_{\greenJ} J\).
Then we define \(R_{\geq J,y}(j)\) as the evaluation of \(w[j+1,k]\) (notice that this is \emph{not} \(y\cdot w[j+1,k]\)).
Similarly, let \(k\) be the least position such that \(w[k,j-1]\cdot y \geq_{\greenJ} J\).
Then we write \(L_{\geq J,y}(j)\) for the evaluation of \(w[k,j-1] \).
When \(y=1\), we simply write \(R_{\geq J}\) and \(L_{\geq J}\).
We also define \(\bar{R}_{\geq J,y}\) and \(\bar{L}_{\geq J,y}\) for the variants that include \(j\) in the intervals.
Note that it is not clear how to compute one from the other without access to successor predicates.

\begin{definition}
  Let \(M\) be a monoid.
  The schema $\schema_M$ is defined as the following set of relations, where \(x,y\in M\) and \(J\) is a \(\greenJ\)-class of \(M\):
  \begin{itemize}
    \item \(R_{\geq J,y,x}\) that contains \(j\) if and only if \(R_{\geq J,y}(j) = x\),
    \item \(L_{\geq J,y,x}\) that contains \(j\) if and only if \(L_{\geq J,y}(j) = x\),
    \item their variants \(\bar{R}_{\geq J,y,x}\) and \(\bar{L}_{\geq J,y,x}\).
  \end{itemize}
\end{definition}
These unary relations satisfy conditions (i) and (ii) of \cref{claim:st_relations_exists}, see the full version. %
The idea to compute the evaluation of an infix \(w]j,k[\) with these relations is to use them to compute the decomposition of \(w]j,n]\) given by~\cref{fact:J_falling}.
This can be done by existentially quantifying the positions \(l_{0},\ldots,l_{m}\) and checking the properties of the decomposition using uniquely universal quantifiers.
The first \(l_{s}\) after \(k\) allows to deduce the \(\greenR\)-class of \(w]j,k[\): indeed, \(w]j,l_{s}[\) is \(\greenJ\)-equivalent to \(w]j,k[\) by the definition of the decomposition, which is strengthened into \(\greenR\)-equivalence because one is a prefix of the other.
Using the second decomposition of ~\cref{fact:J_falling}, from rigth to left, we also deduce the \(\greenL\)-class of \(w]j,k[\).
This identify the \(\greenH\)-class of the infix, then more technical work is needed to compute precisely the evaluation.

\section{\texorpdfstring{The regular languages of \(\UDynProp\)}{The regular languages of UDynProp}}
\label{sec:udynprop}
If binary auxiliary relations are allowed, the quantifier-free fragment \DynProp of \DynFO is powerful enough to maintain all regular languages \cite{GeladeMS12}. 
In this section, we provide a characterization of the regular languages that can be maintained using quantifier-free update formulas if only unary auxiliary relations are permitted. 
These are precisely the regular languages that are recognized by a group.
Intuitively, a computation in a group can be reversed at any moment using the presence of inverses.
That is, if a computation on \(xy\) ends in some memory state, it is possible to retrieve the memory state after the computation of \(x\).
We denote by \(\bG\) this class of languages.

\begin{theorem}[Restatement of \cref{thm:thm2}]\label{thm:udynprop:groups}
	For any regular language $L$:
	\begin{itemize}
	 \item[] $\Member{L}  \in \UDynProp$ if and only if the syntactic monoid of $L$ is a group.
	\end{itemize}

\end{theorem}

We first show that every regular group language can be maintained in \UDynProp, using a simple adaptation of the dynamic program that maintains all regular languages using binary auxiliary relations.
Then we prove that other regular languages cannot be maintained. For this, we adapt a lower bound result from \cite{zeume_substructure_2015} to show that \UDynProp cannot maintain the regular language ``there is at least one \regex{a}'' and then generalize this result to all regular languages whose syntactic monoid is not a group.

\subsection{Maintaining regular group languages with quantifier-free formulas}

We show that unary auxiliary relations and quantifier-free update formulas are sufficient to maintain the evaluation problems of groups. The upper bound from Theorem~\ref{thm:udynprop:groups} follows.

\begin{lemma}\label{lem:groups_in_udynprop}
	If \(G\) is a group, then the problems \(\Member{G}\), \(\StrictPrefix{G}\) and \(\StrictSuffix{G}\) are in \(\UDynProp\).
\end{lemma}

\begin{proofsketch}
Gelade et al. \cite{GeladeMS12} explain how to maintain membership, strict prefixes and suffixes for all regular languages with binary auxiliary relations. Their approach can directly be used to maintain these problems for monoids. For every monoid element $g$, the corresponding update formulas use a bit $Q_g$ to indicate whether the whole input evaluates to $g$ and three further auxiliary relations:
 \begin{itemize}
 	\item a unary relation $P_g$ that contains all positions $i$ such that the prefix ending at position $i-1$ evaluates to $g$,
    \item a unary relation $S_g$ that contains all positions $i$ such that the suffix starting at position $i+1$ evaluates to $g$, and
    \item a binary relation $I_g$ that contains all pairs $(i,j)$ of positions such that the infix starting at position $i+1$ and ending at position $j-1$ evaluates to $g$.
 \end{itemize}
 
 Assume that the element at position $p$ is set to the monoid element $g_1$. Let $i, j$ be positions with $i < p-1$ and $j > p+1$. 
 After the change, the infix from position $i+1$ to position $j-1$ evaluates to the element $g_0 \cdot g_1 \cdot g_2$, where $g_0$ is the result of evaluating the infix from position $i+1$ to position $p-1$ and $g_2$ is the result of evaluating the infix from position $p+1$ to position $j-1$. Both elements $g_0$ and $g_2$ can be read from the auxiliary relations of the form $I_g$ without using quantifiers, given the position $p$ of the change and the positions $i, j$ for which the auxiliary relations are updated.
 
 If the monoid is a group, we can replace atoms $I_g(i,j)$ in the update formulas as follows.
 Let $g_w$ be the result of evaluating the whole input and let $g_p, g_s$ be the evaluations of the prefix up to position $i-1$ and the suffix starting from position $j+1$, respectively. Let $g_i$ and $g_j$ be the elements at position $i$ and $j$. All of these elements can be obtained from the input and auxiliary relations without using quantifiers, given $i$ and $j$. 
 From $g_w = g_p \cdot g_i \cdot g \cdot g_j \cdot g_s$, for the evaluation $g$ of the infix from position $i+1$ to position $j-1$, we obtain $g$ as $g = g_i^{-1} \cdot g_p^{-1} \cdot g_w \cdot g_s^{-1} \cdot g_j^{-1}$, as every group element has an inverse.
 
 So, we replace any atom $I_g(i,j)$ by the formula $\bigvee_{g = g_i^{-1} \cdot g_p^{-1} \cdot g_w \cdot g_s^{-1} \cdot g_j^{-1}} Q_w \wedge W_{g_i}(i) \wedge W_{g_j}(j) \wedge P_{g_p}(i) \wedge S_{g_s}(j)$ and obtain update formulas that only use at most unary auxiliary relations. The result follows.
\end{proofsketch}

\subsection{Lower bound for regular non-group languages}

We want to show that Lemma~\ref{lem:groups_in_udynprop} is optimal and any language whose syntactic monoid is not a group cannot be maintained in \UDynProp.
The proof is in three steps. 
First, we show that a specific regular language cannot be maintained in \UDynProp. 
We generalize this and show that \UDynProp cannot maintain the membership problem for any monoid that is not a group. At last, we infer that one cannot maintain any regular language whose syntactic monoid is not a group by exploiting known closure properties of this class of languages.

\subsubsection{A regular language that is not in \UDynProp}

Schwentick and Zeume~\cite{zeume_substructure_2015} have shown that the graph reachability problem for so-called 1-layer graphs --- edges may only exist from the start node $s$ to some set $A$ of nodes and from this set $A$ to the target node $t$ --- is not in \UDynProp.
Very similarly, one can show that one cannot maintain in \UDynProp whether a set is empty, where changes add elements to the set or remove elements from it. 
We adapt the proof to show that the regular language $\regex{(a+b)^{*}a(a+b)^{*}}$, that is, the language of all words over the alphabet $\{\regex{a},\regex{b}\}$ that contain an $\regex{a}$, is also not in \UDynProp.
The subtle difference between these problems is that the positions of a word are linearly ordered, but the elements of a set are not. However, despite this additional structure, basically the same proof goes through.
We provide some more detail, also because we want to adapt the proof techniques in Section~\ref{sec:udynsop}.

As proof technique we use the Substructure Lemma \cite{GeladeMS12}, which basically says that if two substructures including their auxiliary relations are isomorphic and subject to changes that respect that isomorphism, then also the auxiliary relations that result from applying quantifier-free update formulas are isomorphic.

We make this more formal. For any natural number $n$, we write \([n]\) for the set \(\{1,\ldots,n\}\). 
If $w$ is a word of some length $n$, $I$ is a subset of $[n]$ and $\calA$ is a set of auxiliary relations, we write $(w, \calA)_{|I}$ to denote the restriction of $(w, \calA)$ to the positions $I$. 
Let $w, w'$ be words of length $n \leq m$ respectively and let $\pi$ be a mapping from $[n]$ to $[m]$. 
We call $\pi$ \emph{order-preserving} if $i < i'$ implies $\pi(i) < \pi(i')$.
Two sequences $\alpha = \delta_{1} \cdots \delta_{\ell}$ and $\alpha'= \delta'_{1} \cdots \delta'_{\ell}$ of changes of $w$ and $w'$, respectively, are \emph{$\pi$-respecting} if for every pair $\delta_i, \delta'_i$ of changes it holds that if $\delta_i = \set_\sigma(j)$ for some symbol $\sigma$ and some position $j$, then $\delta'_i = \set_\sigma(\pi(j))$.

If $\alpha$ is a sequence of changes and $\Pi$ is a dynamic program, we denote by $\Pi_{\alpha}(w,\calA)$ the pair of word and auxiliary relations that is obtained by applying the changes $\alpha$ and the corresponding update formulas from $\Pi$ to the word $w$ and the auxiliary relations $\calA$.

\begin{lemma}[Substructure Lemma~\protect{\cite[Lemma 1]{GeladeMS12}}]
	\label{lem:substructure}
	Let $\prog$ be a dynamic \UDynProp program and let $w\in \Sigma^{n}$ and $w'\in \Sigma^{m}$ be two words over the alphabet $\Sigma$. Further, let $\calA$ and $\calA'$ be sets of at most unary auxiliary relations over the domain of $w$ and $w'$ and the schema of $\prog$.
	Assume there are \(I\subseteq [n]\) and \(I'\subseteq [m]\) such that $(w,\calA)_{|I}$ and $(w',\calA')_{|I'}$ are isomorphic via some order-preserving isomorphism \(\pi\).
	Then $\pi$ is also an order-preserving isomorphism from \(\Pi_{\alpha}(w,\calA)_{|I}\) to \(\Pi_{\alpha'}(w',\calA')_{|I'}\), for any \(\pi\)-respecting sequences \(\alpha\) and \(\alpha'\) of changes.
\end{lemma}

To obtain a lower bound via the Substructure Lemma, we need to find isomorphic substructures such that two sequences of isomorphism-respecting changes lead to words that differ with respect to membership in a regular language.
To find such structures, we use Higman's Lemma.

\begin{lemma}[{see \cite[Chapter 6]{Lothaire_1997}}]
	\label{lem:higman}
	Let \((w_{i},\cA_{i})_{i\geq 1}\) be a sequence of words \(w_{i}\in\Sigma^{i}\) and sets \(\calA_{i}\) of relations over a unary schema and over the domain of \(w_{i}\).
	There exists two integers \(n<m\), a subset \(I\subseteq [m]\) and a mapping \(\pi: [n]\rightarrow I\) such $\pi$ is a order-preserving isomorphism from \(w_{n}\) to \(w_{m|I}\).
\end{lemma}

We now show the desired result. 

\begin{lemma}[{c.f.~\cite[Proposition~4.8]{zeume_substructure_2015}}]
	\label{lem:udynprop_in_group}
	The language $\regex{(a+b)^{*}a(a+b)^*}$ is not in \(\UDynProp\).
\end{lemma}
\begin{proof}
	Assume towards a contradiction that there is a dynamic \UDynProp program $\prog$ that maintains  \(\regex{(a+b)^{*}a(a+b)^*}\).
	We consider the sequence $(w_i, \calA_i)_{i \geq 1}$, where $w_i$ is the word $\regex{a}^i$ and $\calA_i$ are auxiliary relations that $\prog$ obtains when starting from an empty input word of size $i$ and setting all positions to $\regex{a}$.
	By~\cref{lem:higman}, there are \(n<m\), \(I\subseteq [m]\) and \(\pi\) such that $\pi$ is an order-preserving isomorphism from $(w_n, \calA_n)$ to $(w_m, \calA_m)_{|I}$.
	
	The sequences $\alpha = \set_{b}(1) \cdots \set_{b}(n)$ and $\alpha' = \set_{b}(\pi(1)) \cdots \set_{b}(\pi(n))$ are $\pi$ respecting, so according to Lemma~\ref{lem:substructure}, the structures $\Pi_{\alpha}(w_n, \calA_n)$ and $\Pi_{\alpha'}(w_m, \calA_m)_{|I}$ are also isomorphic, which means that $\prog$ gives the same answer regarding whether the words are in the language $\regex{(a+b)^{*}a(a+b)^*}$.
	But the first word only consists of $\regex{b}$ after the changes are applied, so does not belong to the language, while the second word still contains at least one $\regex{a}$, contradicting the assumption that $\prog$ maintains the language.
\end{proof}

\subsubsection{\UDynProp cannot maintain any regular non-group language}

We now lift the unmaintainability result of Lemma~\ref{lem:udynprop_in_group} to all regular languages whose syntactic monoid is not a group.
First, we observe that \UDynProp cannot maintain the membership problem of any monoid that is not a group.

\begin{lemma}\label{lem:udynprop:nonongroups}
	 If $M$ is a monoid that is not a group, then $\Member{M}$ is not in \UDynProp.
\end{lemma}
\begin{proof}
	The syntactic monoid $U_1$ of the language $\regex{(a+b)^{*}a(a+b)^*}$ has only the elements $1$ and $a$, where $1$ is the neutral element and $a$ acts as a zero: $1 \cdot x = x \cdot 1 = x$ and $a \cdot x = x \cdot a = a$ for every $x \in \{1,a\}$. As $a$ has no inverse, $U_1$ is not a group. 
	It follows from Fact~\ref{claim:monoid_to_language} and Lemma~\ref{lem:udynprop_in_group} that $\Member{U_1}$ is not in \UDynProp.
	
	Now assume that for some arbitrary non-group monoid $M$ the problem $\Member{M}$ is in \UDynProp. 
	We first claim that $U_1$ is a submonoid of $M$.
	To see that, observe that $M$ must have an element $x$ whose idempotent power is not the identity. 
	Otherwise, if $x^\omega = 1$ for every element $x \in M$, then let $n > 0$ be chosen such that $x^n = 1$ is the idempotent. Then $x \cdot x^{n-1} = x^{n-1} \cdot x = 1$ holds, so $x$ has the inverse $x^{n-1}$, so $M$ is a group, contradicting the assumption.
	Therefore, $M$ has an idempotent $e$ that is different from $1$ and $\{1,e\}$ is a submonoid of $M$ that is isomorphic to $U_1$.
	
	As $U_1$ is a submonoid of $M$ and in particular divides $M$, assuming $\Member{M} \in \UDynProp$ leads to a contradiction, as from Fact~\ref{fact:closure_division} the false statement $\Member{U_1} \in \UDynProp$ would follow.
\end{proof}

Now we have to translate Lemma~\ref{lem:udynprop:nonongroups} from monoids to languages, that is, we need a converse to Fact~\ref{claim:monoid_to_language}. So, we have to prove that if \UDynProp can maintain some regular language, it can also maintain the membership problem of its syntactic monoid.

We employ the theory of (pseudo)varieties, which goes back to Eilenberg~\cite{eilenberg_bookB_1976}.

Let \(\calL\) be a class of regular languages over the alphabet $\Sigma$.
It is \emph{closed under Boolean operations} if \(L_{1}\cup L_{2}\), \(L_{1}\cap L_{2}\) and \(L_{1}^{c}\) are in \(\calL\) for any \(L_{1}, L_{2}\in \calL\).
It is \emph{closed under quotients} if \(\sigma^{-1} L =\{w\ |\ \sigma w\in L \}\) and \(L \sigma^{-1}  =\{w\ |\ w\sigma \in L \}\) are in \(\calL\) for any \(L\in\calL\) and \(\sigma\in\Sigma\).
It is \emph{closed under inverse morphisms} if \(\mu^{-1}(L)\) is in \(\calL\) for any \(L\in \calL\), where \(\Gamma\) is any alphabet and \(\mu:\Gamma \rightarrow \Sigma^{*}\) is a function that is extended to \(\Gamma^{*}\) letter-wise.

A class \(\calV\) of regular languages is a \emph{(pseudo)variety} if it is closed under Boolean operations, quotients and inverse morphisms.
We often omit the ``pseudo'' and simply call such classes ``varieties''. For this work, the following result is used.

\begin{lemma}[\protect{\cite{eilenberg_bookB_1976}}, see~\protect{\cite[Theorem XIII.4.10]{pin_book_2014}}]
	\label{thm:eilenberg}
	Let \(\calV\) be a variety of regular languages and let $M$ be the syntactic monoid of some language \(L\in \calV\).
	Then any language recognized by \(M\) is also in \(\calV\).
\end{lemma}

This implies that the converse of~\cref{claim:monoid_to_language} holds for \UDynProp.

\begin{fact}
	\label{claim:language_to_monoid}
	Let \(\calC\) be a class of formulas such that the regular languages $L$ with \(\Member{L} \in \UDynC\) form a variety.
	Let \(L\) be a regular language and \(M\) its syntactic monoid:
	\[\Member{L} \in \UDynC \Rightarrow \Member{M} \in\UDynC. \]
\end{fact}

\begin{proof}
	Assume that \(\Member{L}\) is in \(\UDynC\) and that $M$ is its syntactic monoid. 
	For each $x \in M$, the monoid $M$ recognizes the language of all sequences of monoid elements that evaluate to $x$, so by Lemma~\ref{thm:eilenberg}, this language is also in  \(\UDynC\).
	A dynamic $\UDynC$ program for \(\Member{M}\) maintains these languages for each $x \in M$.
\end{proof}

The ``only if'' direction of Theorem~\ref{thm:udynprop:groups} follows: observe that the regular languages of \UDynProp form a variety, as \UDynProp is trivially closed under Boolean operations, but also closed under quotients and inverse morphisms due to~\cref{lem:inv_morphism_quotients}.
Therefore, by Fact~\ref{claim:language_to_monoid}, if \UDynProp can maintain a regular language, it can also maintain the membership problem of its corresponding syntactic monoid. But, by Lemma~\ref{lem:udynprop:nonongroups}, \UDynProp cannot maintain this problem for any monoid that is not a group.

\section{\texorpdfstring{The regular languages of \(\UDynSop\)}{The regular languages of UDynSigma-1+}}
\label{sec:udynsop}

In this section, we provide a characterization of the regular languages that can be maintained using positive $\Sigma_1$-formulas if only unary auxiliary relations are permitted. Here, positive $\Sigma_1$-formulas are $\exists^*$-formulas that may not use negations in the quantifier-free part, with the only exception that negations can be used directly in front of the built-in linear order $\leq$. %

An inspection of the proof of Theorem \ref{thm:udynprop:groups} shows that positive quantifier-free formulas and therefore in particular $\Sigma_1^+$-formulas can maintain all regular languages recognized by groups. It is also known that $\Sigma_1^+$-formulas can express all regular languages recognized by the class $\bJ^{+}$ of ordered monoids (see Section \ref{section:ordered-monoids} for a definition). It turns out that languages maintainable in $\UDynSop$ are exactly those recognizable by ordered monoids from the wreath product $\bJ^{+} \ast \bG$ of $\bJ^{+}$ and $\bG$. Intuitively, this class contains monoids that first do a ``group computation'' followed by a ``$\bJ^{+}$-computation'' (see Section \ref{section:ordered-monoids} for a formalization).   

\begin{theorem}[Restatement of \cref{thm:thm3}]\label{theorem:udynsigmaone:j-wreath-groups}
	For any regular language $L$:
	\[\text{$\Member{L}  \in \UDynSop$ if and only if the ordered syntactic monoid of $L$ is in $\bJ^{+} \ast \bG$.}\]

\end{theorem}

We will show, in Section \ref{section:udynsigmaone:upper}, that every regular language recognized with $\bJ^{+} \ast \bG$ is in $\UDynSop$. Then, in Section \ref{section:udynsigmaone:lower}, we prove that other regular languages cannot be maintained. To this end, we first  show  that the language $\regex{b^*}$ cannot be maintained, using an adaptation of the substructure lemma to $\UDynSop$. Finally, relying on a characterization of $\bJ^{+} \ast \bG$, we show that if another language could be maintained, then so could the language $\regex{b^*}$. We therefore look at monoids supplemented with additional ordering information.

\subsection{Ordered monoids and wreath products}
\label{section:ordered-monoids}

For the characterization obtained in Section \ref{sec:udynprop}, we exploited that the class of languages maintainable in $\UDynProp$ is a variety and can therefore be studied by looking at the syntactic monoids of its languages. Unfortunately, the class \(\UDynSop\) of languages is not a variety, as it is not closed under complementation: we will later see that the language \(\regex{(a+b)^{*}a(a+b)^{*}}\) can be maintained,  but its complement \(\regex{b^{*}}\) cannot. As both languages are recognized by the same monoid $U_1 = \{1, a\}$, the information provided by syntactic monoids alone is not sufficient for understanding $\UDynSop$.

\subparagraph*{Ordered monoids} An algebraic theory for classes of regular languages satisfying all the requirements of a variety but complementation has been developed by Pin~\cite{pin_variety_1995}. In this theory, syntactic monoids are supplemented by an additional order. An \emph{ordered monoid} is a pair \((M,\leq)\) where \(M\) is a monoid and \(\leq\) is a partial order on \(M\) compatible with the operation, that is, \(xx'\leq yy'\) whenever \(x\leq y\) and \(x'\leq y'\). We write $\up{x}$ for the set of all elements $y$ with $y \geq x$.  We call \(P \subseteq M\) an \emph{upset} if it is closed under $\up$, that is, if $\up{x} \subseteq P$ for all $x \in P$. A \emph{morphism} from \((M,\leq)\) to \((N,\leq)\) is a monoid morphism from \(M\) to \(N\) such that for all \(x,y\in M\), \(x\leq y\) implies \(\varphi(x)\leq \varphi(y)\). The notions \emph{quotient}, \emph{submonoid} and \emph{division} for monoids can be transferred to ordered monoids using this new notion of morphism. 

A language \(L\) is \emph{recognized} by \((M,\leq)\) if there is a (monoid) morphism \(\varphi:\Sigma^{*}\rightarrow M\) and an upset \(P\) such that \(L=\varphi^{-1}(P)\).
The \emph{syntactic ordered monoid} \((M,\leq)\) for a regular language \(L\) consists of the transition monoid $M$ of the minimal automaton and an order $\leq$ such that $\delta_1 \leq \delta_2$ for $\delta_1, \delta_2 \in M$ if for some $\delta_3 \in M$ and any state $q$ of the minimal automaton it holds that if $\delta_3(\delta_1(q))$ is an accepting state, so is $\delta_3(\delta_2(q))$.

\begin{example}
  Recall that \(U_{1} = \{1,a\}\) is the syntactic monoid of \(\regex{(a+b)^{*}a(a+b)^{*}}\), and therefore also of its complement \(\regex{b^{*}}\).
  However, these languages can be distinguished by their syntactic ordered monoids.
  Let \(U_{1}^{+}\) (resp. \(U_{1}^{-}\)) be the ordered monoid \(U_{1}\) equipped with the order \(1\leq a\) (resp. \(1\geq a\)).
  Then \(U_{1}^{+}\) is the syntactic ordered monoid of \(\regex{(a+b)^{*}a(a+b)^{*}}\), whereas \(U_{1}^{-}\) is the syntactic ordered monoid of \(\regex{b^{*}}\).
  Note that \(\regex{b^*}\) cannot be recognized by \(U_{1}^{+}\) because it is the inverse image of \(1\), which is not an upset.
\end{example}

A prominent example of a class of languages defined via ordered monoids is \(\bJ^{+}\).
\begin{definition}
We denote by \(\bJ^{+}\) the class of all languages whose syntactic ordered monoid has the property that \(1\leq x\) for every \(x\in M\).
\end{definition}
It is known\footnote{In that reference, all orders are reversed. In particular, recognition is defined in terms of \emph{downsets}.}~\cite[Theorem~3.4]{diekert_logic_2008} that this class coincides with the class of regular languages expressible by \(\Sigma_{1}\)-formulas. A simple inspection of the proof gives that all formulas are furthermore \(\Sigma_{1}^{+}\)-formulas.

\subparagraph*{Ordered wreath products} A successful application of wreath products has been to decompose languages into simpler ones, e.g. in the celebrated decomposition theorem of Krohn and Rhodes~\cite{krohn_rhodes_theorem_1965}. One intuition for these products comes from the cascading of finite state automata. Suppose that \(\calA_{1}\) is a deterministic finite state automaton over \(\Sigma\) with state set \(Q\), and \(\calA_{2}\) is a deterministic finite state automaton over \(Q\times \Sigma\). On input \(w\), their cascade product \(\calA_{2} \circ \calA_{1}\) first annotates \(w\) with the states reached in its run in \(\calA_{1}\). Then the resulting enhanced word is fed into \(\calA_{2}\) to check for acceptance.
 
 The algebraic counterpart of this cascade product is the semidirect product of monoids. Here, we use the version for ordered monoids introduced by Pin and Weil~\cite{pin_semidirect_02}. Let \((M,\leq)\) and \((N,\leq)\) be two ordered monoids. To ease notation, we make the common notational convention to denote the operation of \(M\) additively.\footnote{Note that it does not mean that \(M\) is commutative.} The semidirect product of  \((M,\leq)\) and \((N,\leq)\) is defined with respect to a given left action. A \emph{left action} $\cdot$ of \((N,\leq)\) on \((M,\leq)\) is a map \((y,x)\mapsto y\cdot x\) from \(N\times M\) to \(M\) such that for every \(x,x_{1}x_{2}\in M\) and \(y,y_{1},y_{2}\in N\) satisfies the axioms (1) \((y_{1}y_{2})\cdot x = y_{1} \cdot (y_{2}\cdot x)\), (2) \(y\cdot (x_{1}+x_{2}) = y\cdot x_{1} + y\cdot x_{2}\), (3) \(1\cdot x = x\), (4) \(y\cdot 1=1\), (5) if \(x_{1}\leq x_{2}\) then \(y\cdot x_{1}\leq y\cdot x_{2}\); and (6) if \(y_{1}\leq y_{2}\) then \(y_{1}\cdot x\leq y_{2}\cdot x\).

 The \emph{semidirect product} \((M,\leq) \ast (N,\leq)\) of \((M,\leq)\) and \((N,\leq)\) with respect to the left action~$\cdot$ is the ordered monoid on \(M\times N\) defined
 by the operation \((x_{1},y_{1})(x_{2},y_{2}) = (x_{1} + y_{1}\cdot x_{2}, y_{1}y_{2})\) for any \(x_{1},x_{2}\in M\) and \(y_{1},y_{2}\in N\).
 The order on the product is defined componentwise, that is \((x_{1},y_{1})\leq (x_{2},y_{2})\) if and only if \(x_{1}\leq x_{2}\) and \(y_{1}\leq y_{2}\).

Unfolding the definition for a longer product sheds light on the connection between the semidirect product on monoids and the cascaded product on automata. For \(x_{1},\ldots,x_{n}\in M\) and \(y_{1},\ldots,y_{n}\in N\), the product in \(M\ast N\) is (by using the axioms defining a left action):
\begin{equation}
  \label{eq:unfold_wreath}
 (x_{1},y_{1})\cdots (x_{n},y_{n}) = (x_{1} + y_{1}\cdot x_{2} + (y_{1}y_{2})\cdot x_{3} + \cdots + (y_{1}\cdots y_{n-1})\cdot x_{n}  , y_{1}\cdots y_{n} )
\end{equation}

Indeed, in a cascade product, the transition function when reading a letter in the first automaton changes according to the transition function realized so far by the second automaton.
When reading the \(i^{\text{th}}\) letter, the annotation is determined by \(y_{1}\cdots y_{i-1}\) and the action \((y_{1}\cdots y_{i-1})\cdot x_{i}\) gives the transition function that is used.

The \emph{wreath product} \(\bV \ast \bW\) of two classes \(\bV\) and \(\bW\) of ordered monoids contains all ordered monoids that divide a semidirect product of the form \((M,\leq)\ast (N,\leq)\) for \((M,\leq)\in \bV\) and \((N,\leq)\in \bW\).

\begin{remark}
  We chose this definition of wreath products for simplicity.
  While it already exists in the literature, a most standard choice for \(\bV \ast \bW\) would have been to take all ordered monoids dividing \(M^{N}\times N\) endowed with some specific order and operation, when \((M,\leq)\in \bV\) and \((N,\leq)\in \bW\).
  In fact, both definitions are equivalent when both \(\bV\) and \(\bW\) are closed under direct products.
  This is the case for the classes under study here, that is for \(\bJ^{+}\) and \(\bG\).
  Furthermore, the only result on wreath products that is used out-of-the-shelf is Theorem 4.11 in~\cite{pin_semidirect_02} (stated as~\cref{thm:EJp} here).
  Therein, the equivalence between all definitions is proved (see their Proposition 3.5, our definition is (1) and the alternative one is (2)).
\end{remark}

\subsection{\texorpdfstring{Maintaining regular languages in $\bJ^{+} \ast \bG$ with positive $\Sigma_1$-formulas}{Maintaining regular languages in J+ * G with positive Sigma-1-formulas}}
\label{section:udynsigmaone:upper}

We will now show that unary auxiliary relations and positive $\Sigma_1$-formulas are sufficient to maintain the evaluation problems of ordered monoids in $\bJ^{+} \ast \bG$. The upper bound from Theorem \ref{theorem:udynsigmaone:j-wreath-groups} then follows from the following \cref{claim:ordered_monoid_to_language} which is the analogon of \cref{claim:monoid_to_language} for ordered monoids.
For an ordered monoid \((M,\leq)\),  define \(\Member{M,\leq}\) to be the problem that, for a word $w$ subject to changes, maintain a dedicated bit \(q_{x}\) that holds \(1\) whenever $w \in \up{x}$, for every \(x\in M\).

\begin{restatable}{fact}{orderedMonoidLanguages}
  \label{claim:ordered_monoid_to_language}
  Let \(\calC\) be a class of formulas closed under \(\lor\). Further let \(L\) be a regular language and \((M,\leq)\) its syntactic ordered monoid. Then: 
  \[\Member{M,\leq } \in \UDynC \Rightarrow \Member{L} \in\UDynC. \]
\end{restatable}

\begin{lemma}\label{lemma:udynsigmaone:j-wreath-groups:upper}
  If \((M,\le) \in \bJ^{+}\ast \bG\), then \(\Member{M,\leq}\) is in \(\UDynSop\).
\end{lemma}

\begin{proofsketch}
  We show how to maintain \(\Member{M,\leq}\) when \((M,\leq)= (J,\leq) \ast (G,\leq)\) for a fixed left action, where \((J,\leq)\in \bJ^{+}\) and \(G\) is a group.
  The result then follows from~\cref{fact:closure_division}, that can be easily extended to ordered monoids.
  
  Let \(w=(x_{1},g_{1})\cdots(x_{n},g_{n})\) be the input word, where \(x_{i}\in J\) and \(g_{i}\in G\) for all $i \leq n$.
  Let $P = \up{(p,q)} = (\up{p}, \up{q})$ be an upset of \((M,\leq)\), where \(p\in J\) and \(q\in G\).
  We describe how to maintain whether \(w\) evaluates in \(P\).
  
  By~\cref{lem:groups_in_udynprop}, we can maintain all evaluations \(g_{1}\cdots g_{i-1}\) of strict prefixes of the projection of $w$ to the component from $G$ in \UDynProp, as well as the information whether $g_1 \cdots g_n$ evaluates in $\up{q}$. 
  The corresponding update formulas are actually positive: the original update formulas in \cite{GeladeMS12} are positive and the adaptations of the proof of~\cref{lem:groups_in_udynprop} do not introduce negations. 
  This allows to maintain the values \((g_{1}\cdots g_{i-1})\cdot x_{i} \in J \) using positive quantifier-free formulas.
  
  Using the Equation~\ref{eq:unfold_wreath}, we still need to determine whether \(x_{1} + g_{1}\cdot x_{2} + (g_{1}g_{2})\cdot x_{3} + \cdots + (g_{1}\cdots g_{n-1})\cdot x_{n}\) evaluates in \(\uparrow p\).
  As the membership problem of $(J, \leq)$ can be expressed by a \(\Sigma_{1}^{+}\) formula ~\cite[Theorem~3]{diekert_logic_2008}, \(\Sigma_{1}^{+}\) update formulas can perform that evaluation.
\end{proofsketch}

\subsection{\texorpdfstring{Lower bound for regular languages not in $\bJ^{+} \ast \bG$}{Lower bound for regular languages not in J+ * G}} 
\label{section:udynsigmaone:lower}

We want to show that Lemma \ref{lemma:udynsigmaone:j-wreath-groups:upper} is optimal and any language whose ordered syntactic monoid is not in $\bJ^{+} \ast \bG$ cannot be maintained in $\UDynSop$. The proof is in two steps. First, we show that the regular language $\regex{b^*}$ cannot be maintained in $\UDynSop$. We then introduce a characterization of $\bJ^{+} \ast \bG$ that implies that if a language not recognized by $\bJ^{+} \ast \bG$ could be maintained, then so could the language $\regex{b^*}$. 

\subsubsection{\texorpdfstring{A regular language that is not in $\UDynSop$}{A regular language that is not in UDynSigma-1+}}

We show that the language $b^*$ is not in $\UDynSop$ by generalizing the substructure lemma from $\DynProp$ to $\UDynSop$ and then applying it to $b^*$.

The substructure lemma for $\UDynProp$ relies on the inability of quantifier-free formulas to access elements outside the isomorphic substructures. We adopt the lemma to allow existential quantification in update formulas by weakening the requirement on the substructures to be isomorphic. To this end,  let $(w, \mathcal A), (w', \mathcal A')$ be words of length $n \leq m$ respectively annotated with sets of unary auxiliary relations $ \mathcal A,  \mathcal A'$.
For \(i\in[n]\), the \emph{type} of $(w, \mathcal A)_{|i}$ consists of \(w_{i}\) and of the set of all relations of \(\mathcal A\) containing \(i\).
Let $\pi$ be a mapping from $[n]$ to $[m]$. We say that $\pi$ is \emph{type-monotonic} if the type of $(w, \mathcal A)_{|i}$ is a subset of the type of $(w', \mathcal A')_{|\pi(i)}$ for all $i \in [n]$, i.e. if whenever $R{^\mathcal A}(i)$ holds then so does $R{^{\mathcal{A'}}}(\pi(i))$.
In this case, we say that $(w,\calA)$ and $(w',\calA')_{|\pi([n])}$ are type-monotonic via \(\pi\).

\begin{lemma}[Substructure lemma for \(\UDynSop\)]
  \label{lem:extended_substructure}

  	Let $\prog$ be a dynamic $\UDynSop$-program and let $w\in \Sigma^{n}$ and $w'\in \Sigma^{m}$ be two words over the alphabet $\Sigma$. Further, let $\calA$ and $\calA'$ be sets of at most unary auxiliary relations over the domain of $w$ and $w'$ and the schema of $\prog$.
  	
	Assume there is \(I\subseteq [m]\) such that $(w,\calA)$ and $(w',\calA')_{|I}$ are type-monotonic via some order-preserving mapping \(\pi\).
	Then $\pi$ is also a type-monotonic mapping from \(\Pi_{\alpha}(w,\calA)\) to \(\Pi_{\alpha'}(w',\calA')_{|I}\), for any \(\pi\)-respecting sequences \(\alpha\) and \(\alpha'\) of changes.

\end{lemma}

\begin{proof}
  We show the statement for a single change \(\alpha=(\set_{\sigma}(i))\) and \(\alpha'=(\set_{\sigma}(\pi(i))\) with \(i\in [n]\); the general statement follows by induction on the length of change sequences.

	For showing that \(\Pi_{\alpha}(w,\calA)\) and \(\Pi_{\alpha'}(w',\calA')_{|I}\) are type-monotonic via $\pi$, consider a unary relation symbol $R$ updated by formula \(\varphi(x;y)\df \exists z_1 \dots \exists z_k \psi \in\Sigma_{1}^{+} \) when a change to \(\sigma\) is made. Suppose $\varphi(j;i)$ evaluates to true in $(w,\calA)$ by choosing $z_1, \dots, z_k$ as $j_1,  \ldots, j_k$. Then $\varphi(\pi(j);\pi(i))$ evaluates to true in $(w', \calA')$ by choosing  $z_1, \dots, z_k$ as $\pi(j_1),  \ldots, \pi(j_k)$, because $(w,\calA)$ and $(w',\calA')_{|I}$ are type-monotonic via \(\pi\) and  $\psi$ only contains positive atoms.
\end{proof}

To conclude, we apply Higman's lemma in a similar fashion as in~\cref{sec:udynprop}.

\begin{lemma}\label{lemma:b_star_lower_bound}
  The language \(\regex{b^*}\) is not in \(\UDynSop\).
\end{lemma}
\begin{proof}
	Assume towards a contradiction that there is a dynamic \(\UDynSop\)-program $\prog$ that maintains  \(\regex{b^*}\).
	We consider the sequence $(w_i, \calA_i)_{i \geq 1}$, where $w_i$ is the word $\regex{a}^i$ and $\calA_i$ are the auxiliary relations that $\prog$ obtains when starting from an empty input word of size $i$ and setting all positions to $\regex{a}$.
	By~\cref{lem:higman}, there are \(n<m\), \(I\subseteq [m]\) and \(\pi\) such that $\pi$ is an order-preserving isomorphism from $(w_n, \calA_n)$ to $(w_m, \calA_m)_{|I}$. This implies, in particular, that $\pi$ is type-monotonic on $(w_n, \calA_n)$ and $(w_m, \calA_m)_{|I}$.

	The sequences $\alpha = \set_{b}(1) \cdots \set_{b}(n)$ and $\alpha' = \set_{b}(\pi(1)) \cdots \set_{b}(\pi(n))$ are $\pi$ respecting, so according to Lemma~\ref{lem:extended_substructure}, the structures $\Pi_{\alpha}(w_n, \calA_n)$ and $\Pi_{\alpha'}(w_m, \calA_m)_{|I}$ are also type-monotonic. But then, because the first word belongs to the language after applying the changes, the program must also say so for the second word after applying the changes (as it must be monotonic on the answer bit). This is a contradiction, because the second word does not belong to the language after the changes.
\end{proof}

\subsubsection{\texorpdfstring{$\UDynSop$ cannot maintain any regular non-$\bJ^{+} \ast \bG$  language}{UDynSigma-1+ cannot maintain any regular non-J+ * G  language}}

To lift the unmaintainability result of Lemma~\ref{lemma:b_star_lower_bound} to all regular languages whose syntactic monoid is not in $\bJ^{+} \ast \bG$, we first show that $\UDynSop$ cannot maintain the membership problem of such monoids and then lift this to all languages recognized by them.

We have seen that the definition of \(\bJ^{+} \ast \bG\) with the wreath product is useful to design algorithms using only formulas of low complexity and to obtain upper bounds. However, for understanding which monoids do not belong to the class, an equivalent characterization due to Pin and Weil~\cite{pin_semidirect_02} is helpful. Denote by \(\bE\bJ^{+}\) the class\footnote{Note Pin and Weil denote \(\bE\bJ^{+}\) by \(\bB\bG^{+}\) and reverse the order.} of languages whose syntactic ordered monoids satisfy \(1\leq e\) for every idempotent \(e\).
For instance, \(\regex{(a+b)^{*}aa(a+b)^{*}}\) is not in \(\bJ^{+}\) but in \(\bE\bJ^{+}\), while \(\regex{(ab)^{*}}\) is in neither.
Furthermore, every group is in \(\bE\bJ^{+}\) as they have a single idempotent which is the identity.

\begin{theorem}[Pin and Weil \protect{\cite[Thm 4.11]{pin_semidirect_02}}]
  \label{thm:EJp}
  \(\bJ^{+} \ast \bG = \bE\bJ^{+}\)
\end{theorem}

This characterization immediately yields \(U_{1}^{-}\) as witness for non-membership in $\bJ^{+} \ast \bG$.
\begin{lemma}\label{fact:uonep_smallest}
  If \((M,\leq ) \notin \bJ^{+} \ast \bG\), then \(U_{1}^{-}\) divides \((M,\leq )\).
\end{lemma}

\begin{proof}
  By definition of \(\bE\bJ^{+}\), we can find an idempotent \(e\in M\) such that \(1\not\leq e\).
  It implies in particular that \(e\) is different from the identity.
  We consider the submonoid \((U_{1},\leq)\) of \((M,\leq)\) with $U_1 =\{1,e\}$ and do a case distinction on $\leq$:
  \begin{itemize}
    \item If \(1\geq e\), then \((U_{1},\leq)\) is exactly \(U_{1}^{-}\).
    \item If \(1\) and \(e\) are incomparable, then the order \(\leq\) restricted to \(U_{1}\) is the equality \(=\).
          We conclude by remarking that \(U_{1}^{-}\) is a quotient of \((U_{1},=)\), using the identity function.
  \end{itemize}
\end{proof}

\begin{lemma}\label{lem:udynsop:notinJG}
	 If $(M, \leq)$ is an ordered monoid that is not in $\bJ^{+} \ast \bG$, then $\Member{M, \leq}$ is not in $\UDynSop$.
\end{lemma}
\begin{proof}
	We first observe that $\UDynSop$ cannot maintain $\Member{U_1^-}$, due to Fact~\ref{claim:ordered_monoid_to_language} and Lemma~\ref{lem:udynprop_in_group} and because $U_1^-$ is the ordered syntactic monoid of the language $\regex{b^*}$.
	
	Now assume, towards a contradiction, that $\Member{M, <}$ is in $\UDynSop$ for some monoid $(M, <)$ which is not in $\bJ^{+} \ast \bG$. By \cref{fact:uonep_smallest}, \(U_{1}^{-}\) is a submonoid of \((M,\leq)\) and hence \(\Member{U_{1}^{-}}\) is in \(\UDynSop\) by~\cref{fact:closure_division} straightforwardly extended to ordered monoids (it suffices to remark that the inverse image of an upset by a  morphism is an upset), contradicting our observation from above.
\end{proof}

Now we have to translate the previous lemma from ordered monoids to languages, that is, we need a converse to Fact~\ref{claim:ordered_monoid_to_language}. To this end, we prove that if $\UDynSop$ can maintain a regular language, it can also maintain the membership problem of its ordered syntactic monoid.

To this end, we use a variant of varieties that does not require closure under complementation. A class \(\calV\) of regular languages is a \emph{positive (pseudo)variety} if it is closed under positive Boolean operations, quotients and inverse morphisms.

Examples of positive varieties are \(\bV\ast \bW\) whenever \(\bV\) and \(\bW\) are positive varieties~\cite[Proposition 3.5]{pin_semidirect_02}. Also \(\UDynSop\) is a positive variety, because it is closed under positive Boolean operations (because \(\Sigma_{1}^{+}\) is closed under \(\lor\) and \(\land\)) and under quotients and inverse morphisms (due to~\cref{lem:inv_morphism_quotients}). The theorem of Eilenberg translates to positive varieties:

\begin{theorem}[Pin \protect{\cite[Theorem 4.12]{pin_book_2014}}]
  \label{thm:positive_eilenberg}
  Let \(\calV\) be a positive variety of regular languages and let \((M,\leq)\) be the syntactic ordered monoid of a language $L \in \calV$. Then any language recognized by \((M,\leq)\) is also in \(\calV\).
\end{theorem}

This implies that the converse of~\cref{claim:ordered_monoid_to_language} holds for positive varieties.

 \begin{fact}
   \label{claim:language_to_ordered_monoid}
   Let \(\calC\) be a class of formulas such that the regular languages of \(\UDynC\) form a positive variety. Let \(L\) be a regular language and \((M,\leq)\) its syntactic ordered monoid:
   \[\Member{L} \in \UDynC \Rightarrow \Member{M,\leq} \in\UDynC. \]
 \end{fact}

\begin{proof}
  Assume that \(\Member{L}\) is in \(\UDynC\).
  Let \(\varphi: M^{*} \rightarrow M\) be the morphism that evaluates a word in \(M^{*}\), that is that maps \(x_{1}\cdots x_{n}\) to the product of the \(x_{i}\) for \(1\leq i \leq n\).
  For any fixed \(x\in M\), notice that \(\uparrow x\) is an upset of \(M\).
  Thus \(\Member{\varphi^{-1}(\uparrow x)}\) is in \(\UDynC\) for every \(x\in M\), by~\cref{thm:positive_eilenberg}.
  A dynamic program for \(\Member{M}\) uses each of these programs to generate one bit of output.
\end{proof}

The ``only if'' direction of Theorem~\ref{theorem:udynsigmaone:j-wreath-groups} now follows. By Fact~\ref{claim:language_to_ordered_monoid} and because the regular languages of $\UDynSop$ form a positive variety: if $\UDynSop$ can maintain a regular language, it can also maintain the membership problem of its ordered syntactic monoid. But, by Lemma~\ref{lem:udynsop:notinJG}, $\UDynSop$ cannot maintain this problem for any ordered monoid not in \(\bJ^{+}\ast \bG\).

\section{Conclusion}
\label{sec:conclusion}

We characterized the regular languages of $\UDynSop$ and $\UDynProp$ by properties of their (ordered) syntactic monoids. The most important fragment of $\DynFO$ for which a characterization remains open is $\UDynSo$ (see also the discussion in the full version of this paper).%

\begin{openquestion*}
	Which regular languages are in $\UDynSo$, i.e. can be  maintained by $\exists^*$ update formulas and unary auxiliary relations?
\end{openquestion*}

Two obstacles to an algebraic characterization of $\UDynSo$ are (1) that $\Sigma_1$ is not closed under composition and thus $\UDynSo$ is a priori not a (positive) variety; and (2) an absence of lower bounds techniques against \(\UDynSo\). One path towards resolving (1) is via considering  positive lp-varieties, which only require closure under positive Boolean operations and under quotients of inverses of length-preserving morphisms~\cite{Pin_cvar_2010}. For (2), new lower bound techniques beyond the substructure lemma (and its variant used for $\UDynSop$) seem to be necessary. 

We shortly discuss the expressive power of $\UDynSo$. Consider the class $\bJ$ of languages that can be described by a Boolean combination of \(\Sigma_{1}\)-formulas. This class is equal to the class of piecewise-testable languages~\cite{simon_j_1975}.
We can show that \(\UDynSo\) is strictly more powerful than \(\UDynSop\), by looking at the unordered counterpart of \(\bJ^{+} \ast \bG\).

\begin{restatable}{lemma}{JinUDynSO}
  If \(L\) is a regular language in \(\bJ\ast\bG\), then \(\Member{L}\) is in \(\UDynSo\).
\end{restatable}

Unfortunately, \(\bJ\ast\bG\) does not provide an exact characterization, because the language \(\regex{\Sigma^{*}aa\Sigma^{*}}\) is in \(\UDynSo\), but its ordered syntactic monoid is not in \(\bJ\ast\bG\).

\begin{restatable}{conjecture}{DynSoSigmatG}
  Let \(\cV\) be the positive lp-variety of regular languages \(L\) such that \(\Member{L}\) is in \(\UDynSo\).
  Then \(\bJ\ast \bG \subsetneq \cV \subseteq \Sigma_{2}\ast \bG\).
  In particular, \(\regex{ (\Sigma^{*}aa\Sigma^{*})^{c}}\) is not\footnote{Note that non-membership in \(\Sigma_{2}\ast \bG\) can be checked with the MeSCaL software~\cite{place_mescal_2025}.} in \(\cV\).
\end{restatable}

Another direction is to explore (fragments of) the class of context-free languages.  It is known that every context-free language can be maintained in $\DynFO$ \cite{GeladeMS12}, yet the known dynamic program requires $\Sigma_1$-update formulas and $4$-ary auxiliary relations.

\begin{openquestion*}
	Which dynamic resources (quantifier-alterations, arity of auxiliary relations,\dots) are required for maintaining context-free languages?
\end{openquestion*}

\bibliography{bib}

@article{diekert_logic_2008,
  author       = {Volker Diekert and
                  Paul Gastin and
                  Manfred Kufleitner},
  title        = {A Survey on Small Fragments of First-Order Logic over Finite Words},
  journal      = {Int. J. Found. Comput. Sci.},
  volume       = {19},
  number       = {3},
  pages        = {513--548},
  year         = {2008},
  url          = {https://doi.org/10.1142/S0129054108005802},
  doi          = {10.1142/S0129054108005802},
  bibsource    = {dblp computer science bibliography, https://dblp.org}
}

@inproceedings{place_mescal_2025,
  author       = {Thomas Place and
                  Marc Zeitoun},
  editor       = {Giuseppa Castiglione and
                  Sabrina Mantaci},
  title        = {A First Taste of MeSCaL, a Tool for Solving Membership Problems for
                  Regular Languages},
  booktitle    = {Implementation and Application of Automata - 29th International Conference,
                  {CIAA} 2025, Palermo, Italy, September 22-25, 2025, Proceedings},
  series       = {Lecture Notes in Computer Science},
  volume       = {15981},
  pages        = {281--298},
  publisher    = {Springer},
  year         = {2025},
  url          = {https://doi.org/10.1007/978-3-032-02602-6\_20},
  doi          = {10.1007/978-3-032-02602-6\_20},
  bibsource    = {dblp computer science bibliography, https://dblp.org}
}

@INPROCEEDINGS{place_groups_2019,
  author={Place, Thomas and Zeitoun, Marc},
  booktitle={2019 34th Annual ACM/IEEE Symposium on Logic in Computer Science (LICS)},
  title={Separation and covering for group based concatenation hierarchies},
  year={2019},
  volume={},
  number={},
  pages={1-13},
  doi={10.1109/LICS.2019.8785655}}

@inproceedings{pin_bridges_98,
author = {Pin, Jean-Eric},
title = {Bridges for Concatenation Hierarchies},
year = {1998},
isbn = {3540647813},
publisher = {Springer-Verlag},
address = {Berlin, Heidelberg},
booktitle = {Proceedings of the 25th International Colloquium on Automata, Languages and Programming},
pages = {431–442},
numpages = {12},
series = {ICALP '98}
}

@article{Pin_cvar_2010,
author = {Pin, Jean-Éric and Straubing, Howard},
journal = {RAIRO - Theoretical Informatics and Applications},
language = {eng},
month = {3},
number = {1},
pages = {239-262},
publisher = {EDP Sciences},
title = {Some results on C-varieties},
url = {http://eudml.org/doc/92759},
volume = {39},
year = {2010},
}

@book{mcnaughton_fo_star_free_1971,
author = {McNaughton, Robert and Papert, Seymour A.},
title = {Counter-Free Automata (M.I.T. research monograph no. 65)},
year = {1971},
isbn = {0262130769},
publisher = {The MIT Press}
}

@article{buchi_mso_reg_1960,
  title={Weak Second‐Order Arithmetic and Finite Automata},
  author={J. B{\"u}chi},
  journal={Mathematical Logic Quarterly},
  year={1960},
  volume={6},
  pages={66-92},
  doi = {10.1007/978-1-4613-8928-6_22}
}

@InProceedings{simon_j_1975,
author="Simon, Imre",
editor="Brakhage, H.",
title="Piecewise testable events",
booktitle="Automata Theory and Formal Languages",
year="1975",
publisher="Springer Berlin Heidelberg",
address="Berlin, Heidelberg",
pages="214--222",
isbn="978-3-540-37923-2",
doi={10.1007/3-540-07407-4_23}
}

@article{schutzenberger_aperiodic_1965,
title = {On finite monoids having only trivial subgroups},
journal = {Information and Control},
volume = {8},
number = {2},
pages = {190-194},
year = {1965},
issn = {0019-9958},
doi = {10.1016/S0019-9958(65)90108-7},
url = {https://www.sciencedirect.com/science/article/pii/S0019995865901087},
author = {M.P. Schützenberger}
}

@article{Schutzenberger_syntactic_1955,
author = {Schützenberger, M. P.},
journal = {Séminaire Dubreil. Algèbre et théorie des nombres},
language = {fre},
pages = {1-24},
publisher = {Secrétariat mathématique},
title = {Une théorie algébrique du codage},
url = {http://eudml.org/doc/111094},
volume = {9},
year = {1955},
}

@book{Lothaire_1997,
  author={Lothaire, M.},
 place={Cambridge},
 edition={2},
 series={Cambridge Mathematical Library},
 title={Combinatorics on Words},
 publisher={Cambridge University Press}, year={1997},
 collection={Cambridge Mathematical Library}}

@article{pin_semidirect_02,
  title={Semidirect products of ordered semigroups},
  author={Pin, Jean-Eric and Weil, Pascal},
  journal={Communications in Algebra},
  volume={30},
  number={1},
  pages={149--169},
  year={2002},
  publisher={Taylor \& Francis}
}

@article{pin_variety_1995,
  title={A variety theorem without complementation},
  author={Pin, Jean-{\'E}ric},
  journal={Russian Mathematics (Izvestija vuzov. Matematika)},
  volume={39},
  pages={80--90},
  year={1995}
}

@article{green_relations_51,
 ISSN = {0003486X, 19398980},
 URL = {http://www.jstor.org/stable/1969317},
 author = {J. A. Green},
 journal = {Annals of Mathematics},
 number = {1},
 pages = {163--172},
 publisher = {[Annals of Mathematics, Trustees of Princeton University on Behalf of the Annals of Mathematics, Mathematics Department, Princeton University]},
 title = {On the Structure of Semigroups},
 urldate = {2025-07-10},
 volume = {54},
 year = {1951}
}

@article{krohn_rhodes_theorem_1965,
 ISSN = {00029947},
 URL = {http://www.jstor.org/stable/1994127},
 author = {Kenneth Krohn and John Rhodes},
 journal = {Transactions of the American Mathematical Society},
 pages = {450--464},
 publisher = {American Mathematical Society},
 title = {Algebraic Theory of Machines. I. Prime Decomposition Theorem for Finite Semigroups and Machines},
 urldate = {2024-02-20},
 volume = {116},
 year = {1965}
}

@article{pin95,
  title={{BG=PG}: a success story},
  author={Pin, Jean-Eric},
  journal={NATO ASI Series C Mathematical and Physical Sciences-Advanced Study Institute},
  volume={466},
  pages={33--48},
  year={1995},
  publisher={Dordrecht, Holland, Boston, D. Reidel Pub. Co.}
}

@article{Frandsen_dynamic_1997,
author = {Skovbjerg Frandsen, Gudmund and Miltersen, Peter Bro and Skyum, Sven},
title = {Dynamic word problems},
year = {1997},
issue_date = {March 1997},
publisher = {Association for Computing Machinery},
address = {New York, NY, USA},
volume = {44},
number = {2},
issn = {0004-5411},
url = {10.1145/256303.256309},
doi = {10.1145/256303.256309},
journal = {J. ACM},
pages = {257–271},
numpages = {15}
}

@unpublished{pin_book_2014,
author = {Pin, Jean-Eric},
year = {2014},
title = {Mathematical foundations of automata theory},
url = {http://www.irif.fr/~jep/PDF/MPRI/MPRI.pdf}
}

@book{eilenberg_bookB_1976,
  author={Samuel Eilenberg},
  title={Automata, Languages and Machines, Vol. B},
  year={1976},
  publisher={Academic Press Inc},
  url={https://www.sciencedirect.com/bookseries/pure-and-applied-mathematics/vol/59/part/PB}
}

@phdthesis{hesse_phd_2003,
  title        = {Dynamic Computational Complexity},
  author       = {Hesse, William},
  year         = 2003,
  school       = {University of Massachusetts Amherst},
}

@article{GeladeMS12,
author = {Gelade, Wouter and Marquardt, Marcel and Schwentick, Thomas},
title = {The dynamic complexity of formal languages},
year = {2012},
issue_date = {August 2012},
publisher = {Association for Computing Machinery},
address = {New York, NY, USA},
volume = {13},
number = {3},
issn = {1529-3785},
url = {10.1145/2287718.2287719},
doi = {10.1145/2287718.2287719},
journal = {ACM Trans. Comput. Logic},
articleno = {19},
numpages = {36}
}

@phdthesis{Zeume15thesis,
  author       = {Thomas Zeume},
  title        = {Small dynamic complexity classes},
  school       = {Technical University Dortmund, Germany},
  year         = {2015},
  url          = {https://hdl.handle.net/2003/34163},
  urn          = {urn:nbn:de:101:1-201606013164},
  bibsource    = {dblp computer science bibliography, https://dblp.org}
}

@article{zeume_substructure_2015,
title = {On the quantifier-free dynamic complexity of Reachability},
journal = {Information and Computation},
volume = {240},
pages = {108-129},
year = {2015},
note = {MFCS 2013},
issn = {0890-5401},
doi = {10.1016/j.ic.2014.09.011},
url = {https://www.sciencedirect.com/science/article/pii/S0890540114001205},
author = {Thomas Zeume and Thomas Schwentick},
}

@article{PatnaikI97,
	Author = {Sushant Patnaik and Neil Immerman},
	Bibsource = {dblp computer science bibliography, http://dblp.org},
	Doi = {10.1006/jcss.1997.1520},
	Journal = {J. Comput. Syst. Sci.},
	Number = {2},
	Pages = {199--209},
	Title = {Dyn-{FO}: {A} Parallel, Dynamic Complexity Class},
	Volume = {55},
	Year = {1997}}

\appendix

\section{Proofs from Section~\ref{sec:prelim}}
\invmorphismquotients*
\begin{proof}
	Let $L$ be a language over some alphabet $\Sigma$ and let $\prog$ be a dynamic unary \(\DynC\) program for $\Member{L}$.
	\begin{enumerate}
		\item Fix a mapping $h : \Gamma \to \Sigma^*$ and let $s \df \max_{\gamma \in \Gamma} |h(\gamma)|$ be the maximal size of a string in the image of $h$.
		We construct a dynamic unary \(\DynC\) program $\prog'$ for $\Member{h^{-1}(L)}$ that uses the following auxiliary relations:
		\begin{itemize}
			\item $s$ unary relations $Q_{\sigma,1}, \ldots, Q_{\sigma,s}$ for every $\sigma \in \Sigma$, and
			\item $s$ unary relations $R_1, \ldots, R_s$ for every unary auxiliary relation $R$ of $\prog$.
		\end{itemize}
		The relations of the form $Q_{\sigma, \ell}$ are used to store the word $h(w)$ of size (at most) $sn$, where $w \in \Gamma^*$ is the current input word of $\prog'$ and $n$ is the domain size of the input word. If $i \in Q_{\sigma, \ell}$ then the letter $\sigma$ is at position $(i-1)s + \ell$ of $h(w)$.
		So, we store $h(w)$ in blocks of $s$ positions, one block for each symbol of $w$. If $|h(w_i)|<s$ then intuitively the remaining positions are ``filled'' with $\epsilon$.
		The relations $R_\ell$ directly play the role of the auxiliary relations $R$ on the maintained input word $h(w)$ for $L$.
		
		The auxiliary relations are updated as follows. Suppose a change $\set_\gamma(i)$ occurs. Let $\sigma_1 \cdots \sigma_s \df h(\gamma)$, with $\sigma_j \in \Sigma \cup \{\epsilon\}$, be the block of symbols that changes in the stored word $h(w)$. Intuitively, $\prog'$ simulates $\prog$ for the $s$ corresponding changes. As $s$ is a constant, this is possible by nesting slight modifications of the update formulas of $\prog$.
		
		The update formulas are constructed in $s$ stages. The formulas of the first stage simulate the update after the first change, setting the symbol at position $(i-1)s+1$ to $\sigma_1$, as follows.
		The index $i$ is included into $Q_{\sigma_1,1}$ and removed from any other relation of the form $Q_{\sigma',1}$, which is easy to express by quantifier-free formulas.
		The formula for an auxiliary relation $R_\ell$ is obtained from the update formula $\varphi_\gamma^{R}(x;y)$ from $\prog$ by 
		\begin{itemize}
			\item replacing every atom of the form $S(x)$ by $S_\ell(x)$, as the variable $x$ denotes the position that is updated,
			\item replacing every atom of the form $S(y)$ by $S_1(y)$,  as $y$ denotes the position that is changed,
			\item iteratively replacing every formula $\exists z \psi$ by $\exists z \bigvee_{\ell_z = 1, \ldots, s} \psi_{\ell_z}$, where $\psi_{\ell_z}$ results from $\psi$ by replacing every atom of the form $S(z)$ by $S_{\ell_z}(z)$,
			\item an atom $z \leq z'$ is replaced by $z < z'$ if $\ell_z > \ell_{z'}$, and otherwise remains unchanged.
		\end{itemize}

		For the remaining $s-1$ stages, the formulas are obtained analogously, but any reference to an atom $R_\ell$ or $Q_{\sigma, \ell}$ is replaced by the formula defining this relation from the previous stage.
		\item 
		We construct a dynamic unary \DynC program $\Pi'$ for $\Member{L \sigma^{-1}}$, the argument for $\Member{\sigma^{-1}L}$ is analogous. The construction is very similar to (a), but simpler. The dynamic program $\Pi'$ uses the following auxiliary relations:
		\begin{itemize}
			\item two unary relations $Q_{\tau,1}$ and $Q_{\tau,2}$ for every $\tau\in\Sigma$,
			\item two unary relations $R_1$ and $R_2$ for every unary auxiliary relation $R$ of $\Pi$, and
			\item an additional unary relation $\textit{Max}$.
		\end{itemize}
		The relations $Q_{\tau,1}$ and $Q_{\tau,2}$ encode a word with which $\Pi$ is simulated as in Part (a). The relations $Q_{\tau,1}$ will be used to store $w$. The relations $Q_{\tau,2}$ will only be used to store the last $\sigma$.
		
		The relation $Q_{\sigma,2}$ is initialized to $\{n\}$, which means that initially, the symbol $\sigma$ is written to the last position of the maintained word.
		Also, $\textit{Max}$ is initialized to $\{n\}$.
		The relations \(R_{1}\) and \(R_{2}\) are then initialized according to the update formulas of \(\Pi\) for the case that the last letter is set to \(\sigma\). This does not require to existentially quantify that position, as any atom $x \leq y$ of the update formula involving the variable $x$ of the changed position can be replaced using $\textit{Max}$ and any atom $R_i(x), Q_{\tau,i}(x)$ can be replaced either by $\top$ or $\bot$.
		
		On a change $\set_\gamma(i)$, the index $i$ is included into $Q_{\gamma,1}$.
		The relations $R_1$ and $R_2$ are maintained as in (a). As a change of $w$ only leads to one change in the word represented by the relations, $Q_{\tau,i}$, it is not necessary to compose formulas. \qedhere
	\end{enumerate}  %
\end{proof}

\monoidtolanguage*
\begin{proof}
	Assume that \(\Member{M}\) is in \(\DynC\) and that \(L\) is recognised by some subset \(P \subseteq M\) with morphism \(\varphi\).
	The dynamic program \(\Pi\) for \(\Member{M}\) stores for each \(x\in M\) whether the maintained word evaluates to \(x\) in a bit \(q_{x}\).
	We can maintain \(\Member{L}\) using an output bit \(b\) that is updated via the disjunction of the update formulas for the \(q_{x}\) with \(x\in P\).
	This new program witnesses \(\Member{L}\in\DynC\) thanks to the closure under \(\lor\) of~\(\calC\).
\end{proof}

\closuredivision*
\begin{proof}
	Take \(M\) such that \(\Member{M}\in \DynC\). We only have to prove the result when  \(N\) is either a submonoid or a quotient of \(M\).
	
	First, let \(N\) be a submonoid of \(M\). %
	It means that up to relabelling of the elements of \(M\), \(N\) can be seen as a subset of \(M\).
	Therefore we have directly an algorithm for \(\Member{N}\) by using the algorithm for \(\Member{M}\).
	
	Second, let \(N\) be a quotient of \(M\) and \(\mu:M\rightarrow N\) be a surjective morphism.
	It means that for every \(x\in N\), we can find \(y_{x}\in M\) such that \(\mu(y_{x})=x\).
	We can use the algorithm for \(\Member{M}\) where the input word \(x_{1}\cdots x_{n}\in N^{*}\) is seen as the word \(y_{x_{1}}\cdots y_{x_{n}}\in M^{*}\).
	To do that, we replace every occurence in the formulas of the letter relations \(W_{y_{x}}\) by \(W_{x}\) for \(x\in N\).
	The relations \(W_{y}\) not of this form are replaced by False.
	Finally, the output bit for \(x\in N\) is maintained with the disjunction of all the update formulas for the elements in \(\varphi^{-1}(x)\).
\end{proof}

This proof can be adapted to catch ordered monoids as well.
The only difference is for the closure under quotients, where we have to use the fact than an ordered morphism preserves the order.
Thus \(\varphi^{-1}(\up{x})\) is an upset, and can be maintained with the disjunction of all update formulas for \(\up{y}\) with \(y\in\varphi^{-1}(\up{x})\).

\section{Proofs of Section \ref{sec:udynst}}
\subsection{Proof of Lemma \ref{fact:J_falling}}
\label{appendix:I-chain}

\Jfalling*
\begin{proof}
  We only prove the first part of the statement, the second one being symmetrical.
  We set \(\ell_{0} = j\) and iteratively identify the numbers \(\ell_{s}\) for $s>0$.
  Let \(\ell_{s+1}\)  be the least index greater than \(\ell_{s}\) such that \(w[j,\ell_{s+1}-1]\ \cancel{\greenJ}\ w[j,\ell_{s+1}]\).
  Property ii) is satisfied by definition.
  By maximality, we have that \(w[j,i-1]\ \greenJ\ w[j,i]\) for every \(\ell_{s}< i<\ell_{s+1}\).
  We obtain by induction that \( w[j,\ell_{s}]\ \greenJ\ w[j,\ell_{s}+1]\ \greenJ \ \cdots \ \greenJ\ w[j,\ell_{s+1}-1]  \), yielding Property i).

  All is left to see is that it terminates in less than \(|M|\) iterations.
  Let \(x_{s}= w[j,\ell_{s+1}-1]\), for each $0 \leq s < m$.
  By definition,  \(x_{s}\) is a prefix of \(x_{s+1}\) and thus \(x_{s} \geq_{\greenJ} x_{s+1} \).
  Then \(x_{s} \ \cancel{\greenJ}\ x_{s}w_{\ell_{s+1}}\) and it is impossible that \(x_{s+1}\) and \(x_{s}\) are in the same \(\greenJ\)-class.
  It implies that we have a chain $x_{0} >_{\greenJ} x_{1}  >_{\greenJ}\cdots  >_{\greenJ} x_{m-1}$ of strict relations.
  Such a chain has size bounded by the size of \(M\), concluding the proof.
\end{proof}

\subsection{Proof of Lemma \ref{claim:st_relations_exists}}
\label{appendix:sigma2}

\strelationsexists*
 We prove Part ii) in \cref{sec:app:expressibility} and Part i) in \cref{sec:app:maintainablity}
  \subsubsection{\texorpdfstring{Proof that evaluations of infixes can be expressed using the relations in $\schema_M$}{Proof that evaluations of infixes can be expressed using the relations in tau-M}}
\label{sec:app:expressibility}
  We first prove Part ii) of the lemma, that the auxiliary relations from $\schema_M$ are sufficient to express the evaluation of infixes, afterwards we show that the relations can be updated using $\Sigma_2$-formulas.
We will repeatedly use the following fact about strengthening an order relation into an equivalence inside a \(\greenJ\)-class.

\begin{fact}
  \label{claim:R_order_to_R_equiv}
  Let \(x,y\in M\) such that \(x\greenJ y\) and \(x\leq_{\greenR}y\).
  Then \(x\greenR y\).

  Dually, if \(x\greenJ y\) and \(x\leq_{\greenL}y\), then \(x\greenL y\).
\end{fact}

\begin{proof}
  By definition, if \(x \leq_{\greenR} y\) then there is \(\alpha\in M\) such that \(x=y\alpha \).
  With the other assumption, \(y\greenJ y \alpha\) follows.
  Hence there are \(\beta, \gamma\in M\) such that \(y = \beta y\alpha\gamma\).
  Therefore, we also have \(y=\beta^{2}y(\alpha \gamma)^{2}\).
  We can repeat the process until with have an idemptotent on the left and the right: \(y=\beta^{\omega}y(\alpha \gamma)^{\omega}\).
  Thus, \(y\alpha [\gamma(\alpha\gamma)^{\omega-1}] = \beta^{\omega}y(\alpha \gamma)^{2\omega} = y\).
  So \(y\leq_{\greenR}y\alpha =x\), which implies \(x\greenR y\).
\end{proof}

  \paragraph{Evaluation of infixes}

  It is important to understand in which cases the multiplication of two monoid elements produces a fall in the \(\greenJ\)-classes.
  This is given by a first part of Green's lemma.

  \begin{claim}[see \protect{~\cite[Prop V.1.10]{pin_book_2014}}]
    \label{claim:green_lemma_I}
      Let \(x,y\in M\) such that \(x\greenL y\) and they both belong to the \(\greenJ\)-class \(J\).
  Then for any \(z\in M\):
  \[ xz\in J \Leftrightarrow yz \in J .  \]
  \end{claim}

  \begin{proof}
  The statement is symmetrical, we only need to prove one direction.  	
  By assumption, there exists \(\alpha\in M\) such that \(y=\alpha x\).
  Suppose that \(xz\in J\), that is to say \(xz \greenJ x\).
  Because \(xz\leq_{\greenR}x\) is always true, by~\cref{claim:R_order_to_R_equiv} we have \(xz\greenR x\).
  Hence there exists \(\beta\in M\) such that \(xz\beta = x\).
  Thus, \(yz\beta = \alpha x z \beta = \alpha x =y\).
  It implies that \(yz\greenR y\), and in particular that \(yz\in J\), as $\greenR \subseteq \greenJ$.
\end{proof}

This allows to obtain an algebraic characterization to evaluate the \(\greenR\)-classes and \(\greenL\)-classes of infixes.

 \begin{lemma}
   \label{lemma:find_r_class}
   Let \(M\) be a monoid and \(w\in M^{*}\). Let  \(j<k\) be two positions and \(R\) a \(\greenR\)-class.
   Then \(w[j+1,k-1]\) evaluates to an element in \(R\) if and only if there exists \(x_{1},\ldots, x_{m}\in M\) and \(j=\ell_{0}<\ell_{1}<\cdots <k\leq \ell_{m}\) with \(m\leq |M|\) such that:
   \begin{enumerate}[i)]
     \item \(x_{m}\in R\),
     \item for \(1\leq s< m\), \(R_{\geq \greenJ(x_{s})}(j) = x_{s}\),
     \item for \(0\leq s< m\) and \(\ell_{s}< i< \ell_{s+1}\), \(L_{\geq \greenJ(x_{s+1})}(i)\cdot w_{i} \greenJ x_{s+1}\),
     \item for \(1\leq s< m\), \(x_{s}\cdot w_{\ell_{s}} \not\geq_{\greenJ} x_{s} \) and $x_{s} \cdot w_{\ell_s} \greenJ x_{s+1}$.
   \end{enumerate}

   There is a similar statement for \(\greenL\)-classes.
 \end{lemma}

\tikzset{
	subsetneq/.style={
		draw=none,
		edge node={node [sloped, allow upside down, auto=false, fill=white]{$\subsetneq$}}},
	iseq/.style={
		draw=none,
		edge node={node [sloped, allow upside down, auto=false, fill=white, inner sep=2pt]{$\stackrel{?}{=}$}}}
}

\begin{figure}
	\centering
   \begin{tikzpicture}
  	\node (w1) at (0,0) {\LARGE \(w_{j+1}\)};
  	\node (wl1) at (2.5,0) {\LARGE\(w_{\ell_1-1}\)};
  	\node (wl1p) at (3.5,0) {\LARGE\(w_{\ell_1}\)};
  	\node (wl2) at (6,0) {\LARGE\(w_{\ell_2-1}\)};
  	\node (wl2p) at (7,0) {\LARGE\(w_{\ell_2}\)};
  	\node (wl2b) at (10,0) {\LARGE\(w_{k-1}\)};
  	\node (wn) at (11.8,0) {\LARGE\(w_{l_{m}-1}\)};
  	
  	\draw (w1) edge[dotted] (wl1);
  	\draw (wl1p) edge[dotted] (wl2);
  	\draw (wl2p) edge[dotted] (wl2b);
  	\draw (wl2b) edge[dotted] (wn);
  	
  	\draw [decorate,
  	decoration = {calligraphic brace, raise = 0cm, amplitude=0.2cm },
  	line width = 0.2mm] ($(wl1.south east) - (0.1,0)$) --  ($(w1.south west) + (0.1,0)$)
  	node[pos=0.5, black, yshift = -0.4cm] (x1) {\(x_{1}\)};
  	
  	\draw [decorate,
  	decoration = {calligraphic brace, raise = 1cm, amplitude=0.2cm },
  	line width = 0.2mm] ($(wl2.south east) - (0.1,0)$) --  ($(w1.south west) + (0.1,0)$)
  	node[pos=0.5, black, yshift = -1.4cm] (x2) {\(x_{2}\)};
  	
  	\draw [decorate,
  	decoration = {calligraphic brace, raise = 2cm, amplitude=0.2cm },
  	line width = 0.2mm] ($(wn.south east) - (0.1,0)$) --  ($(w1.south west) + (0.1,0)$)
  	node[pos=0.5, black, yshift = -2.4cm] (xn) {\(x_{m}\)};
  	
  	\node (x1p) at (x1 -| x2) {\(x_{1}w_{\ell_{1}}\)};
  	\node (x2p) at (x2 -| xn) {\(x_{2}w_{\ell_{2}}\)};
  	\node (x0) at (x1 -| w1) {\(w_{1}\)};
  	
  	\node at ($(x1.east)!0.6!(x1p.west)$) {\(\textcolor{red}{>_{\mathcal{J}}}\)};
  	\node at ($(x2.east)!0.6!(x2p.west)$) {\(\textcolor{red}{>_{\mathcal{J}}}\)};
  	\node at ($(x0.east)!0.5!(x1.west)$) {\(\textcolor{red}{\mathcal{J}}\)};
  	
  	\node at ($(x1p.south)!0.4!(x2.north)$) {\(\textcolor{red}{\mathcal{J}}\)};
  	\node at ($(x2p.south)!0.2!(xn.north)$) {\(\textcolor{red}{\vdots}\)};

  \end{tikzpicture}
  \caption{Representation of the situation in the proof of~\cref{lemma:find_r_class}}
  \label{fig:falling_J}
\end{figure}

 \begin{proof}
  First, assume that \(w[j+1,k-1]\) evaluates to some \(x\in R\).
  Let \(j=\ell_{0}<\cdots<\ell_{m-1}<k\leq \ell_{m}\) with \(m\leq |M|\) be as given by \cref{fact:J_falling}.
  Note that we have discarded all indices after the first one after \(k\).
  For \(1\leq i\leq m\), let \(x_{i}= w[j+1,\ell_{i}-1]\).
  The situation is depicted in~\cref{fig:falling_J}.
  For i), we observe that \(x_{m}\) and \(x\) are both \(\greenJ\)-equivalent to \(w[j+1,\ell_{m}-1]\) by the choice of the indices $\ell_i$ and the definition of $x_m$. 
  Also by definition, \(x_{m} \leq_{\greenR} x\) holds and thus $x_m \greenR x$ by~\cref{claim:R_order_to_R_equiv}.

  The first property of~\cref{fact:J_falling} asserts that \(w[j+1,l_{s}]\greenJ w[j+1,l_{s+1}-1]\), which is exactly the second part of iv).
  The second property of~\cref{fact:J_falling} gives that for any \(s\), \(x_{s}\) and \(x_{s}\cdot w_{l_{s}}\) are not \(\greenJ\)-equivalent.
  Because \(x_{s}\cdot w_{l_{s}} \leq_{\greenJ} x_{s}\), this implies the first part of iv).
  This also implies that any further prefixes \(w[j+1,i]\) for \(i\geq l_{s}\) cannot be \(\greenJ\)-greater than \(x_{s}\) (indeed, multiplication can only decrease the \(\greenJ\)-class).
  Thus \(w[j+1,l_{s}-1]\) is the greatest prefix starting from \(j+1\) that is \(\greenJ\)-equivalent to \(x_{s}\), yielding ii).
  We now prove that iii) holds. Let \(1\leq s< m\) and \(\ell_{s} < i< \ell_{s+1}\).
  \Cref{fact:J_falling} gives that \(w[j+1,l_s] \greenJ w[j+1,i-1] \greenJ x_{s+1}\).
	Let \(i'\) be the index defining \(L_{\geq \greenJ(x_{s+1})}\) as \(w[i',i-1]\).
	If \(i'>j+1\), then \(w_{i'-1}\cdot w[i',i-1]\) is not \(\greenJ\)-greater than \(x_{s+1}\), by definition of \(L\).
  But \(w_{i'-1}\cdot w[i',i-1]\) is a suffix of \(w[j+1,i-1]\) which is in \(\greenJ(x_{s+1})\).
  This is a contradiction, as multiplying can only decrease the \(\greenJ\)-class. Thus \(i' \leq j+1\).
	By definition, \(w[i',i-1]\) is \(\greenJ\)-greater than \(x_{s+1} \greenJ w[j+1,i-1]\).
  We also have the other direction, as \(w[j+1,i-1]\) is a suffix of \(w[i',i-1]\): they are \(\greenJ\)-equivalent.
	By~\cref{claim:R_order_to_R_equiv}, they are \(\greenL\)-ordered (one is the suffix of the other), hence they are \(\greenL\)-equivalent.
	\Cref{fact:J_falling} gives that \(w[j+1,i-1]w_i J x_{s+1}\), hence we conclude by ~\cref{claim:green_lemma_I}.

  Conversely, assume that Properties i) to iv) are satisfied.
  By induction, we prove that for every \(s\), \(w[j+1,\ell_{s}-1]\) is the greatest prefix of \(w[j+1,|w|]\) that evaluates in \(\greenJ(x_{s})\), and that it evaluates to \(x_{s}\).
  We nest an induction on \(\ell_{s}<i\leq \ell_{s+1}\) to prove that \(w[j+1,i-1]\) belongs to \(\greenJ(x_{s+1})\).
  The case \(i=\ell_{s}+1\) is assumed to hold by the second part of iv).
  Now assume that \(w[j+1,i-1]\) is in \(\greenJ(x_{s+1})\).
  This gives that \(i'\leq j+1\) for \(i'\) the index defining \(L_{\geq \greenJ(x_{s+1})}\) as \(w[i',i-1]\).
	By definition, \(w[i',i-1]\) is \(\greenJ\)-greater than \(x_{s+1}\), which it itself \(\greenJ\)-equivalent to \(w[j+1,i-1]\) by the nested induction hypothesis.
	But also \(w[j+1,i-1]\) is a suffix of \(w[i',i-1]\) and hence the other direction follows.
	Thus \(w[i',i-1]\) and \(w[j+1,i-1]\) are \(\greenJ\)-equivalent, but one in the suffix on the other so there are further \(\greenL\)-equivalent by~\cref{claim:R_order_to_R_equiv}.
  By~\cref{claim:green_lemma_I}, \(L_{\geq \greenJ(x_{s})}(i)\cdot w_{i}\) in \(\greenJ(x_{s})\) implies \(w[j+1,i] \) in \(\greenJ(x_{s})\).
  This achieves the nested induction.
  By iv), \(w[j+1, \ell_{s+1}]\) does not belong to \(\greenJ(x_{s+1})\).
  So by ii), the greatest word greater than \(\greenJ(x_{s+1})\) starting at \(j+1\), which is necessarily \(w[j+1, \ell_{s+1}-1]\), evaluates to \(x_{s+1}\).
  This achieves the induction.
  Because \(x_{m}\) is \(\greenR\)-equivalent with the evaluation of \(w[j+1,k-1]\) (with~\cref{claim:R_order_to_R_equiv}), this concludes the proof.

\end{proof}

This immediately gives a formula to compute infixes.

\begin{corollary}
  \label{cor:r_classes_formulas}
   Let $M$ be a monoid and \(w \in M^*\). 
   For every \(\greenR\)-class \(R\) of \(M\) there is a \(\Sigma_{2}\)-formula \(\psi_{R}(j,k)\) over \(\schema_{M}\) that is satisfied if and only if \(w[j+1,k-1] \) evaluates to an element in \(R\).

   There is a similar statement for \(\greenL\)-classes, giving formulas \(\psi_{L}(j,k)\).
 \end{corollary}

 \begin{proof}
 	The formula $\psi_R$ is a disjunction over all \(x_{1},\ldots,x_{m}\in M\) such that \(x_{m}\in R\). It then existentially quantifies the positions $\ell_0, \ldots, \ell_m$ and checks the following conditions:
   \begin{itemize}
     \item \(j=\ell_{0}<\ell_{1}< \cdots < k \leq \ell_{m}\),
     \item for every \(1\leq s<m\), \(R_{\geq \greenJ(x_{s}), x_{s}}(j)\),
     \item for every \(0\leq s<m\), \(\forall \ell_{s}< i<\ell_{s+1}, \bigvee_{y\cdot z \greenJ(x_{s})} L_{\geq \greenJ(x_{s}),y}(i) \land W_{z}(i)\), %
           \item  for every \(1\leq s<m\), \(   \bigvee_{x_{s}\cdot y \not\geq_{\greenJ} \greenJ(x_{s})} W_{y}(\ell_{s})  \).
   \end{itemize}
   The correctness is given by~\cref{lemma:find_r_class}.
 \end{proof}

 This enables us to find the \(\greenR\)-class and \(\greenL\)-class of every infix, but not the precise element inside the \(\greenH\)-class.
 To obtain it, we need the second part of Green's lemma.

\begin{claim}[see \protect{~\cite[Prop V.1.10]{pin_book_2014}}]
  \label{lem:H_action}
  Let \(x,z\in M\) such that \(x\) and \(xz\) are in the \(\greenH\)-classes \(H_{1}\) and \(H_{2}\), respectively.
  Assuming that \(H_{1}\) and \(H_{2}\) are part of the same \(\greenJ\)-class, the function
  \[ m_{z}: \left\{\begin{array}{rcl}
          H_{1} & \rightarrow & H_{2} \\
          y &\mapsto & yz
        \end{array}\right. \]
  is well-defined and a bijection.
\end{claim}

\begin{proof}
  We note that by~\cref{claim:R_order_to_R_equiv}, because \(xz\leq_{\greenR}x\) is always true, we have that \(xz\greenR x\).

  First of all, we need to show that for every \(y\in H_{1}\), the image \(m_{z}(y)=yz\) is in \(H_{2}\).
  On one hand, \(x\) and \(y\) are in particular \(\greenL\)-equivalent, there are \(\alpha,\beta\in M\) such that \(\alpha x =y\) and \(\beta y=x\).
  Hence \(\alpha xz = yz\) and \(\beta yz = xz\), which means that \(xz \greenL yz\).
  On the other hand, \(yz\leq_{\greenR}y\) and \(yz\greenJ y\) (by the preceding point).
  By~\cref{claim:R_order_to_R_equiv}, we have that \(yz\greenR y\).
  Thus there is a chain of equivalences \(yz\greenR y\greenR x \greenR xz\).
  All together, it stands that \(yz\greenH xz\) and then belongs to \(H_{2}\).

  All is left is to see that \(m_{z}\) is a bijection.
  Because \(xz\greenR x\) there is \(t\in M\) such that \(x = xzt\).
  We consider the function \(m_{t}: H_{2}\rightarrow H_{1}; x \mapsto xt\), which is well-defined with the same reasoning as before.
  For \(y\in H_{1}\), take \(\alpha\) such that \(\alpha x= y\).
  We have \(m_{t}(m_{z}(y))=yzt = \alpha x zt = \alpha x = y\).
  Similarly, for all \(y\in H_{2}\), we have \(m_{z}(m_{t}(y))=y\).
  Hence \(m_{z}\) and \(m_{t}\) are inverses of each other, and \(m_{z}\) is a bijection.
\end{proof}

\begin{lemma}
  \label{lemma:h_class_alg}
  Let \(M\) be a monoid and \(w\in M^{*}\). Let \(j<k\) be two positions and \(H\) a \(\greenH\)-class in a \(\greenJ\)-class \(J\).
  Let \(t\) be any element in \(H\).
  If \(w[j+1,k-1]\in H\) then \(w[j+1,k-1] = m^{-1}_{z}(R_{\geq J}(j))\), where \(z= \bar{R}_{\geq J,t}(k)\) and \(m_{z}\) is the function defined in~\cref{lem:H_action}.
\end{lemma}

\begin{proof}
  Assume that \(R_{\geq J}(j)\) is defined as the evaluation of \(w[j+1,i]\).
  Necessarily, \(i\geq k\) as \(w[j+1,k-1]\) already evaluates in \(J\).
  Let \(w=xyzt\) with \(x=w[1,j]\), \(y=w[j+1,k-1]\), \(z=w[k,i]\) and \(t=w[i+1,|w|]\).
  Thus \(R_{\geq J}(j)=yz\).
  Furthermore, for any \(y'\in H\), in particular for \(y'=y\) or \(y'=t\), \( \bar{R}_{\geq J, y'}(k)=z\).
  Indeed, let \(i'\) be the greatest index such that \(y'\cdot w[k,i']\) evaluates to a \(\greenJ\)-class greater than \(J\).
  By~\cref{lem:H_action}, \(y'z\) and \(yz\) are \(\greenH\)-equivalent and thus \(i'\geq i\).
  Moreover, \(yzw_{i+1}\notin J\) and thus \(y'zw_{i+1}\notin J\) by~\cref{claim:green_lemma_I}, yielding \(i'=i\).

  Finally, \cref{lem:H_action} tells us that \(m_{z}\) is a bijection from \(H\) to the \(\greenH\)-class of \(yz\).
  Therefore, now that we know the values of \(z\) and \(yz\), we can find the desired value of \(y\) with the formula \(y=m_{z}^{-1}(yz)\).
\end{proof}

We finally prove the second Part of~\cref{claim:st_relations_exists}.

\begin{corollary}
  \label{cor:computing_infixes}
   Let $M$ be a monoid and \(w \in M^*\). 
   For every \(x\in M\), there is a \(\Sigma_{2}\)-formula \(\psi_{x}(j,k)\) over \(\schema_{M}\) that is satisfied if and only if \(w[j+1,k-1] \) evaluates to \(x\).
 \end{corollary}

 \begin{proof}

  For \(R\) and \(L\) be the \(\greenR\)-class and \(\greenL\)-class of \(x\), let \(\psi_{R}(j,k)\), \(\psi_{L}(j,k)\) be the formulas given by~\cref{cor:r_classes_formulas}.
  With those, we can determine the \(\greenH\)-class of the evaluation of \(w[j+1,k-1]\): it is \(H=R\cap L\).
  Let \(t\) be an arbitrary element of \(H\).
  To futher determine the element inside that \(\greenH\)-class, we use~\cref{lemma:h_class_alg}.
  This give the formula:
  \[ \psi_{x}(j,k) = \psi_{R}(j,k)  \land \psi_{L}(j,k) \land \bigvee_{x=m_{z}^{-1}(y)} R_{\geq J,y}(j) \land \bar{R}_{\geq J,t,z}(k).\]

  The formula \(\psi_{x}\) is a finite Boolean combination without negations of formulas in \(\Sigma_{2}\) and hence is also in \(\Sigma_{2}\).
 \end{proof}

\subsubsection{\texorpdfstring{Proof that the relations in $\schema_M$ can be maintained}{Proof that the relations in tau-M can be maintained}}
\label{sec:app:maintainablity}
  We now show Part i) of~\cref{claim:st_relations_exists}.
  To that end, we use the formulas computing infixes introduced in the last subsection.

  \begin{lemma}
    For any monoid \(M\), the relations in \(\schema_{M}\) can be maintained with \(\Sigma_{2}\) formulas.
  \end{lemma}

  \begin{proof}
    We only prove it for \(R_{\geq J,y,x}\), for some \(\greenJ\)-class \(J\) of \(M\) and \(x,y\in M\), the case for \(L\) and the barred versions are analogous.
    Let \(\psi_{x}^{\sigma}\) be the formula for infix evaluation after a change as given in the proof of~\cref{thm:reg_in_udynst}.
    Assume an operation \(\set_{\sigma}(i)\) is seen.
  The update formula simply finds the maximal position whose infix is in \(J\):
  \[ \varphi^{R_{\geq J,y',y}}_{\sigma}(j;i) = \exists k> j \Big[\ \bigvee_{\stackrel{y=tz \ y'tz\geq_{\greenJ} J}{y'tzt' \not\geq_{\greenJ} M}} \psi^{\sigma}_{y}(j,k;i) \land W_{t}(j) \land W_{t'}(k) \Big]  . \]
  As only an existential quantifier is used and the $\Sigma_2$-formulas $\psi^{\sigma}_{y}$ are not in the scope of a negation, the full formula is also in \(\Sigma_{2}\).
  \end{proof}

\section{Proofs from Section~\ref{sec:udynsop}}
\orderedMonoidLanguages*
\begin{proof}
  Suppose $\prog$ is a dynamic program for \(M\) that stores whether $w \in \uparrow x$ in an auxiliary bit \(q_{x}\), for all \(x\in M\). Then \(\Member{L}\) can be maintained by adding a new auxiliary bit \(q\), that is updated with the disjunction of the update formulas for the \(q_{x}\) for \(x\in P\).
  Indeed, as \(P\) is an upset, it is the union of \(\uparrow x\) for \(x\in P\).
\end{proof}

\section{\texorpdfstring{Discussion: The regular languages of \(\UDynSo\)}{Discussion: The regular languages of UDynSigma-1}}
\label{sec:udynso}

We discuss approaches towards algebraic characterizations of \(\UDynSo\); expanding upon the discussion in the main part. Two obstacles to an algebraic characterization of $\UDynSo$ are (1) that $\Sigma_1$ is not closed under composition and thus $\UDynSo$ is a priori not a (positive) variety; and (2) an absence of lower bounds techniques against \(\UDynSo\). One path towards resolving (1) is via considering  positive lp-varieties, which only require closure under positive Boolean operations and under quotients of inverses of length-preserving morphisms~\cite{Pin_cvar_2010}. For (2), new lower bound techniques beyond the substructure lemma (and its variant used for $\UDynSop$) seem to be necessary. 

In the rest of this section, we first we discuss what we know about the expressivity of \(\UDynSo\) and then give some ideas of why lp-varieties may be helpful in characterizing~\(\UDynSo\).

\subsection{\texorpdfstring{Expressivity of \(\UDynSo\).}{Expressivity of UDynSigma-1.}}

We shortly discuss the expressive power of $\UDynSo$. Consider the class $\bJ$ of languages that can be described by a Boolean combination of \(\Sigma_{1}\)-formulas, i.e. a Boolean combination of languages of the form \(\Sigma^{*}a_{1}\Sigma^{*}\cdots \Sigma^{*}a_{k}\Sigma^{*}\), with \(a_{1},\ldots,a_{k}\in\Sigma\). This class is equal to the class of piecewise-testable languages~\cite{simon_j_1975}.

We next show that $\UDynSo$ is powerful enough to maintain all regular languages in \(\bJ\ast\bG\). This class can be equivalently characterized as Boolean combination of \emph{group monomials}, i.e. of languages of the form \(L_{1}a_{1}L_{2}\cdots a_{n}L_{n+1}\) where \(a_{1},\ldots,a_{n}\in \Sigma\) and \(L_{1},\ldots,L_{n+1}\) are languages over \(\Sigma^{*}\) recognized by a finite group~\cite[Theorem 6.1]{pin95}.

\JinUDynSO*

\begin{proof}
  We know that \(L\) is a Boolean combination of group monomials.
  We can push the negations to the leaves, and use that \(\UDynSo\) is closed under positive Boolean operations.
  Therefore, we only need to tackle the cases where \(L\) is equal to \(L_{1}a_{1}L_{2}\cdots a_{n}L_{n+1}\) where for \(1\leq i\leq n+1\), \(L_{i}\) is recognized by the group \(G_{i}\); or its complement.
  We start by the first one.
  Let \(G\) be the product of all the \(G_{i}\): It is clear that it recognizes every \(L_{i}\) with some morphism \(\mu\).
  A prefix submonomial of \(L\) is a language of the form \(L_{1}a_{1}\cdots L_{k}a_{k}\mu^{-1}(g)\) for some \(0\leq k\leq n\) and \(g\in G\).
  Similarly, a suffix submonomial of \(L\) is a language of the form \(\mu^{-1}(g) a_{k} L_{k+1}\cdots a_{n}L_{n+1}\) for some \(1\leq k\leq n+1\) and \(g\in G\).
  We use several auxiliary relations:
  \begin{itemize}
    \item \(P_{k,g}\) that contains \(\StrictPrefix{L'}\) for \(L'\) the prefix submonomial of \(L\) associated to \(k\) and \(g\),
    \item \(S_{k,g}\) that contains \(\StrictSuffix{L'}\) for \(L'\) the suffix submonomial of \(L\) associated to \(k\) and \(g\),
  \end{itemize}
  Remark that \(P_{0,g}\) and \(S_{n+1,g}\) contain respectively \(\StrictPrefix{G}\) and \(\StrictSuffix{G}\), and can be maintained with~\cref{lem:groups_in_udynprop}.
  Let \(\set_{a}(i)\) be a change.
  There is a quantifier-free formula \(\varphi^{a}_{g}(j,k;i)\) that expresses whether the value of the new infix \(w_{j+1}\cdots w_{k-1}\) is \(g\).
  It uses \(P_{0,g}\) that is already maintained.
  We describe how to update \(P_{k,g}(j)\) (the other case is symmetrical), for some \(k,j\) and \(g\), under this change.
  This is done with a formula that existentially quantifies over \(1\leq x_{1}<\cdots < x_{k}<j\) and asserts that (1) there is a letter \(a_{l}\) at every position \(x_{l}\), (2) for every \(0\leq l<k\), \(\varphi^{a}_{h}(x_{l},x_{l+1};i)\) holds for some \(h\in\mu(L_{l})\), and (3) \(\varphi^{a}_{g}(x_{k},j;i)\) holds.

  Now the answer bit is maintained via a disjunction over:
  \begin{itemize}
    \item all \(1\leq k\leq n\), and all \(g\in\mu(L_{k})\) and \(g'\in \mu(L_{k+1})\), of \(P_{k-1,g}\land S_{k+1,g'}\),
    \item all \(1\leq k\leq n+1\), and all \(g\cdot \mu(a) \cdot  g'\in \mu(L_{k})\), of \(P_{k,g}\land S_{k,g'}\).
  \end{itemize}
  Indeed, if \(w\) is in \(L\) there is a factorization \(w=w_{1}a_{1}w_{2}\cdots a_{n} w_{n+1}\) with \(w_{k}\in L_{k}\).
  Then either some \(a_{k}\) falls at position \(i\) or not.
  The first case is detected by the first part of the formula and the second case by the second part of the formula,
  as strict prefixes and suffixes from position \(i\) are unaltered by the change.

  Notice that the answer bit is maintained with a quantifier-free formula. Therefore, in the case of the complement of a group monomial, we can simply use the same dynamic program where the formula for the answer bit is negated.
\end{proof}

One might conjecture that \(\bJ\ast\bG\) provides an exact characterization. Unfortunately, this is not the case because the language \(\regex{\Sigma^{*}aa\Sigma^{*}}\) is in \(\UDynSo\), see Example~\ref{example:aa}, but its ordered syntactic monoid is not in \(\bJ\ast\bG\).

We conjecture, however, that no further expressibility is possible beyond  \(\Sigma_{2}\ast \bG\),
where \(\Sigma_{2}\) is the positive variety of ordered monoids recognizing a language definable by a \(\Sigma_{2}[<]\) sentence.

\DynSoSigmatG*

This problem appears challenging.
Lower bounds are scarce, and the class \(\Sigma_{2}\ast\bG\) is poorly understood.
Although it is known to correspond to the level 3/2 of the group-based concatenation hierarchy~\cite{pin_bridges_98}, and it is decidable~\cite{place_groups_2019}, it is unclear if there is a simplification to \(\bE \Sigma_{2}\) as in many other cases.
Moreover, the computational power of \(\UDynSo\) within \(\Sigma_{2}\ast\bG\) seems very subtle, as reflected by the fact that is not a variety but a less studied lp-variety, making it difficult to locate precisely.

\subsection{\texorpdfstring{Varieties tailored for $\UDynSop$}{Varieties tailored for UDynSigma-1+}}
 
As stated above, \(\UDynSo\) is not a priori a variety, since \(\Sigma_{1}\) is not closed under composition.
Consequently, \cref{lem:inv_morphism_quotients} does not yield closure under inverse morphisms.
Nevertheless, the proof idea of the lemma  still yields closure under quotients and also closure under length-preserving morphisms $h \colon \Gamma \to \Sigma$ can be shown. 

This motivates to look for variety theories for these weaker closure properties. It turns out that   such varieties have been explored in \cite{Pin_cvar_2010}. A class of languages is called a \emph{positive lp-variety} if it is closed under positive Boolean operations and under quotients of inverses of length-preserving morphisms.

\begin{fact}
  The class of languages maintainable in \(\UDynSo\) is a positive lp-variety.
\end{fact}

Unlike varieties, lp-varieties may fail to contain two languages with the same syntactic ordered monoid.

\begin{example}
 Consider the class \(\mathcal M\) of regular languages \(L\) whose membership only depends on the length, that is for all \(w,w'\) such that \(|w|=|w'|\) we have \(w\in L\) iff \(w'\in L\).
 It is clearly closed under Boolean operations and quotients.
 As a length-preserving morphism preserves the length of words by definition, \(\mathcal M\) is also closed under length-preserving morphisms.
 However, consider \(L= \regex{((a+b)(a+b))^{*}} \in \mathcal M\) and \(\mu: \{a,b\}^{*}\rightarrow \{a,b\}^{*}\) defined by \(\mu(a)=a\) and \(\mu(b)=aa\).
 Then \(L' = \mu^{-1}(L) = \regex{(b^{*}ab^{*}ab^{*})^{*}}\) is not in \(\mathcal M\).
 This shows that \(\mathcal M\) is not a variety but a lp-variety.

 But \(L\) and \(L'\) both have the group \(\mathbb{Z}_{2}=\{0,1\}\) as syntactic monoid.
 To distinguish them algebraically, we have to look at morphisms recognizing them.
 Indeed, \(L\) is recognized by \(\mu(a)=\mu(b)=1\) and \(L'\) is recognized by \(\mu(a)=1\) and \(\mu(b)=0\).
\end{example}

This motivates the use of a finer syntactic object, which preserves more information about the language.
An ordered stamp is a morphism \(\mu:\Sigma^{*}\rightarrow (M,\leq) \) where \(\Sigma\) is a finite alphabet and \((M,\leq)\) a finite ordered monoid.
Let \(L\) be a regular language and recall that the syntactic ordered monoid \((M,\leq)\) is defined as \(\delta(\Sigma^{*})\), where \(\delta\) is the transition function of the minimal automaton of \(L\).
We defined the syntactic ordered stamp as the morphism \(\mu: \Sigma^{*}\rightarrow (M,\leq)\) such that for every \(a\in\Sigma\), \(\mu(a)=\delta(a)\).
There is a counterpart to~\cref{thm:eilenberg} and~\cref{thm:positive_eilenberg}.

\begin{lemma}[\protect{\cite[Theorem 5.1]{Pin_cvar_2010}}]
  Let \(\cV\) be a positive lp-variety of regular languages and let \(\mu\) be the syntactic ordered stamp of a language \(L\in\cV\).
  Then any language recognized by \(\mu\) is also in \(\cV\).
\end{lemma}

Thus positive lp-varieties together with syntactic ordered stamps may be a good starting point towards characterizing $\UDynSop$.
 
\end{document}